\documentclass[a4paper,10pt]{amsart}

\usepackage[margin=2.5cm]{geometry}
\usepackage{amssymb}
\usepackage{ascmac}
\usepackage[dvipdfmx]{graphicx}
\usepackage{color}
\usepackage{bbm}

\newtheorem{Def}{Definition}[section]
\newtheorem{Thm}[Def]{Theorem}
\newtheorem{Lem}[Def]{Lemma}
\newtheorem{Assumption}[Def]{Assumption}
\newtheorem{Rem}[Def]{Remark}

\numberwithin{equation}{section}

\newcommand{\argmin}{\operatornamewithlimits {argmin}}

\allowdisplaybreaks

\usepackage{ascmac}

\newcommand{\mcc}{\mathcal{C}}
\newcommand{\mcd}{\mathcal{D}}

\newcommand{\mcf}{\mathcal{F}}

\newcommand{\mch}{\mathcal{H}}

\newcommand{\mcj}{\mathcal{J}}
\newcommand{\mcl}{\mathcal{L}}

\newcommand{\mbba}{\mathbb{A}}
\newcommand{\mbbb}{\mathbb{B}}

\newcommand{\mbbh}{\mathbb{H}}
\newcommand{\mbbn}{\mathbb{N}}

\newcommand{\mbbr}{\mathbb{R}}

\newcommand{\mbbi}{\mathbbm{1}} 

\newcommand{\al}{\alpha}
\newcommand{\del}{\delta}
\newcommand{\ep}{\epsilon}

\newcommand{\D}{\Delta}

\newcommand{\Sig}{\Sigma}
\newcommand{\lam}{\lambda}

\newcommand{\gam}{\gamma}
\newcommand{\Gam}{\Gamma}
\newcommand{\p}{\partial}

\newcommand{\cil}{\xrightarrow{\mcl}} 
\newcommand{\cip}{\xrightarrow{p}} 
\newcommand{\argmax}{\mathop{\rm argmax}}
\newcommand{\diag}{\mathop{\rm diag}}

\newcommand{\lf}{\lfloor}
\newcommand{\rf}{\rfloor}

\def\ds#1{\displaystyle{#1}}
\def\nn{\nonumber}
\def\cadlag{c\`adl\`ag}
\def\wp{Wiener process}
\def\cpp{compound Poisson process}

\def\sumj{\sum_{j=1}^{n}}

\def\intj{\int_{t_{j-1}}^{t_j}}

\def\tz{\theta_{0}}

\def\tes{\hat{\theta}_{n}}

\def\aes{\hat{\alpha}_{n}}
\def\bes{\hat{\beta}_{n}}

\def\Xc{X^{\mathrm{cont}}}

\usepackage[textwidth=2cm, textsize=scriptsize]{todonotes}

\title[Estimating diffusion with compound Poisson jumps]
{Estimating diffusion with compound Poisson jumps based on self-normalized residuals}
\date{\today}


\author{Hiroki Masuda}
\address{Faculty of Mathematics, Kyushu University, Motooka 744 Nishi-ku Fukuoka 819-0395, Japan}
\email{hiroki@math.kyushu-u.ac.jp}

\author{Yuma Uehara}
\address{The Institute of Statistical Mathematics, 10-3 Midori-cho, Tachikawa, Tokyo 190-8562, Japan}
\email{y-uehara@ism.ac.jp}

\begin{document}
\setlength{\baselineskip}{4.5mm}

\maketitle

\begin{abstract}
We consider parametric estimation of the continuous part of a class of ergodic diffusions with jumps based on high-frequency samples.
Various papers previously proposed threshold based methods, which enable us to distinguish whether observed increments have jumps or not at each small-time interval, hence to estimate the unknown parameters separately. However, a data-adapted and quantitative choice of the threshold parameter is known to be a subtle and sensitive problem.
In this paper, we present a simple alternative based on the Jarque-Bera normality test for the Euler residuals.
Different from the threshold based method, the proposed method does not require any sensitive fine tuning, hence is of practical value.
It is shown that under suitable conditions the proposed estimator is asymptotically equivalent to an estimator constructed by the unobserved fluctuation of the continuous part of the solution process, hence is asymptotically efficient.
Some numerical experiments are conducted to observe finite-sample performance of the proposed method.
\end{abstract}

\section{Introduction}

Consider the following one-dimensional stochastic differential equation with jumps:
\begin{equation}
dX_t=\left( \sum_{l=1}^{p_\al} \al^{(l)} a^{(l)}(X_t) \right)^{1/2}dw_t+\sum_{k=1}^{p_\beta} \beta^{(k)} b^{(k)}(X_t)dt+c(X_{t-})dJ_t,
\label{hm:sde}
\end{equation}
defined on a complete filtered probability space $(\Omega, \mcf, (\mcf_t)_{t\geq0},P)$.
The ingredients are as follows:
\begin{itemize}
\item The coefficients $\{a^{(l)}(x)\}_{l=1}^{p_\al}$ and $\{b^{(k)}(x)\}_{k=1}^{p_\beta}$ are known measurable functions;
\item The statistical parameter
\begin{equation}
\theta:=(\al,\beta)\in \Theta_{\al} \times \Theta_{\beta}=\Theta
\nonumber
\end{equation}
are unknown, where $\Theta_\al$ and $\Theta_\beta$ are bounded convex domains and subset of $\mbbr^{p_\al}$ and $\mbbr^{p_\beta}$, respectively;
\item $w$ is a standard {\wp} and $J$ a {\cpp} with intensity parameter $\lam\in[0,\infty)$ and i.i.d jump-size random variables $\{\xi_i\}_{i\in\mbbn}$, that is,
\begin{equation*}
J_t=\sum_{i=1}^{N_t} \xi_i;
\end{equation*}
\item $(w,J)$ is $\mcf_t$-adapted, and the initial variable $X_0$ is $\mcf_0$-adapted and independent of $(w,J)$.
\end{itemize}
Throughout this paper, we assume that there exists a true value $\tz:=(\al_0,\beta_0)\in\Theta$.
We want to estimate $\tz$ based on a discrete-time but high-frequency observation $(X_{t^{n}_{j}})_{j=0}^{n}$ from a solution to \eqref{hm:sde}, where the sampling times are supposed to be equally spaced:
\begin{equation}
t^{n}_{j} =jh_{n}
\nn
\end{equation}
for a positive sequence $(h_n)$ such that $h_n \to 0$ and the terminal sampling time $T_{n}:=t^{n}_{n}=nh_n\to\infty$.
Throughout we suppose that $\lam>0$; for diffusion models, many estimator of $\theta$ have been proposed, such as Gaussian quasi-likelihood estimator \cite{Kes97}, adaptive estimator \cite{UchYos12}, multi-step estimator \cite{KamUch15}, to mention few.
The special forms of the coefficients of \eqref{hm:sde} may seem restrictive. However, we are particularly interested in models which can be estimated without heavy computational effort.
As will be mentioned in Section \ref{Asymptotic Results}, we do not need any numerical search of a maximizer to estimate $\theta$ as good as virtual situation where we know every jump instances over $(0,T_n]$.

In the presence of the jump component, elimination of the effect of $J$ is crucial for a reliable estimation of $\theta$. A well-known approach for it is the threshold based method independently proposed in \cite{Man04} and \cite{ShiYos06}; see also \cite{OgiYos11} for subsequent developments. In the method, we look at sizes of the increments
\begin{equation}
\D^n_{j}X=\D_{j}X:=X_{t^n_{j}}-X_{t^n_{j-1}}
\nonumber
\end{equation}
for $j=1,\dots,n$ in absolute value: we assume that one jump has occurred over $(t^n_{j-1},t^n_j]$ if $|\D_{j}X|>r_{n}$ for a pre-specified \textit{jump-detection threshold} $r_{n}>0$, and then estimate $\theta$ after removing such increments.
For a suitably chosen $r_{n}>0$, it is shown that the estimator of $\theta$ is asymptotically normally distributed at the same rate as diffusion models, while finite-sample performance of the threshold method strongly depends on the value of $r_{n}$.
A data-adaptive quantitative choice of $r_n$ is a subtle and sensitive problem in practice;
see \cite{Shi08}, \cite{Shi09}, as well as the references therein.
Obviously, if the model may have ``small'' jumps with positive probability, joint estimation of diffusion and jump components can exhibit a rather bad finite-sample performance; for example, some increments may simultaneously contain small jumps and large fluctuation caused by continuous component.
This practical issue can also be seen in other jump detection methods such as \cite{AitJac09}.

Recently, for estimating the volatility parameter in the non-ergodic framework, i.e., for a fixed $T>0$, $h_n=T/n$ and $T_n\equiv T$, \cite{InaYos18} proposed an alternative estimation procedure called a global jump-detection filter based on the theory of order statistics constructed from the whole increments;
there, it is shown that the global filtering can work both theoretically and numerically better than the previously studied local one (\cite{Man04}, \cite{ShiYos06}, and \cite{OgiYos11}).
Nevertheless, as will be seen later, required conditions on the distribution of jump sizes and decaying rate of $h_n\to 0$ may be more stringent in the case where $T_n\to\infty$. Hence it is not quite clear whether or not and how the global filtering of \cite{InaYos18} is directly applicable to our ergodic setting.


The primary objective of this paper is to formulate an intuitively easy-to-understand strategy, which can simultaneously estimate $\theta$ and detect jumps without any precise calibration of a jump-detection threshold. 
For this purpose, we utilize the approximate self-normalized residuals \cite{Mas13-2}, which makes the classical Jarque-Bera test \cite{JarBer87} adapted to our model.
More specifically, the hypothesis test whose significance level is $\al\in(0,1)$ is constructed by the following manner: let the null hypothesis be of ``no jump component'' against the alternative hypothesis of ``non-trivial jump component'':
\begin{equation}
\mch_0:  {\lam=0} \quad \text{vs} \quad \mch_1: {\lam>0}.
\nn
\end{equation}
Then, if the Jarque-Bera type statistic introduced later is larger than a given percentile of the chi-square distribution with $2$ degrees of freedom,
we reject the null hypothesis $\mch_0$; and otherwise, we accept $\mch_0$.
For such a test, we can intuitively regard that the largest increment contains at least one jump when the null hypothesis is rejected. 
Following this intuition, our proposed method will go as follows:
we iteratively conduct the test with removing the largest increments in the retained samples until rejection of $\mch_0$ is stopped;
after that, we construct the modified estimator of $\theta$ by the remaining samples.
Our method enables us not only just to make a ``pre-cleaning'' of diffusion-like data sequence by removing large jumps which breaks the approximate Gaussianity of the self-normalized residuals, but also to approximately quantify jumps relative to continuous fluctuations in a natural way; see Remark \ref{hm:rem_F.esti}.

This paper is organized as follows: in Section \ref{Preliminaries}, we give a brief summary of the approximate self-normalized residuals, and the Jarque-Bera type test for general jump diffusion models.
Section \ref{Proposed strategy} provides our strategy and some remarks on its practical use. 
In Section \ref{Asymptotic Results}, we will propose a least-squares type estimator and its one-step version for \eqref{hm:sde}.
In the calculation of our estimator we can sidestep optimization, and thus it is numerically tractable, with retaining high representational power of the nonlinearity in the state variable. Moreover, we will prove that our estimator is asymptotically equivalent to the ``oracle'' estimator which is constructed as if we observe the unobserved continuous part of $X$.
We show some numerical experiments results in Section \ref{Numerical experiments}.
Finally, Appendix \ref{hm:sec_proofs} presents the proofs of the results given in Section \ref{Asymptotic Results}.

\medskip

Here are some notations and conventions used throughout this paper.
We largely abbreviate ``$n$'' from the notation like $t_{j}=t^{n}_{j}$ and $h=h_{n}$. 
For any vector variable $x=(x^{(i)})$, we write $\p_x=\left(\frac{\p}{\p x^{(i)}}\right)_i$.
For any process $Y$, $\D_j Y$ denotes the $j$-th increment $Y_{t_{j}}-Y_{t_{j-1}}$.
$C$ denotes a universal positive constant which may vary at each appearance.
$\top$ stands for the transpose operator, and $v^{\otimes2}:= vv^\top$ for any matrix $v$.
The convergences in probability and in distribution are denoted by $\cip$ and $\cil$, respectively.
All limits appearing below are taken for $n\to\infty$ unless otherwise mentioned.
For two nonnegative real sequences $(a_n)$ and $(b_n)$, we write $a_n \lesssim b_n$ if $\limsup_n(a_n/b_n)<\infty$.
For any $x\in\mbbr$, $\lf x \rf$ denotes the maximum integer which does not exceed $x$.

\section{Preliminaries}
\label{Preliminaries}
To see whether a working model fits data well or not, and/or whether data in hand have outliers or not, diagnosis based on residual analysis is often done.
For jump diffusion models, \cite{Mas13-2} formulated a Jarque-Bera normality test based on self-normalized residuals for the driving noise process.
In this section, we briefly review the construction of the self-normalized residual, and the Jarque-Bera statistics with its asymptotic behavior for general ergodic jump diffusion model described as:
\begin{equation}
dX_{t} = a(X_{t},\al)dw_{t} + b(X_{t},\beta)dt + c(X_{t-})dJ_{t}.
\label{yu:sde}
\end{equation}
Given any function $f$ on $\mbbr\times\Theta$ and $s\geq0$, we hereafter write 
\begin{equation*}
f_s(\theta)=f(X_s,\theta),
\end{equation*} 
and in particular, for all $j\in\{0,\dots,n\}$ we denote
\begin{equation}
f_{j}(\theta) = f(X_{t_{j}},\theta).
\nonumber
\end{equation}
For each $j\in\{1,\dots,n\}$, let
\begin{equation}
\ep_j(\al)=\ep_{n,j}(\al):=\frac{\D_{j}X}{\sqrt{a^{2}_{j-1}(\al)h_n}}.
\label{hm:ep_def}
\end{equation}
Then, following \cite{Mas13-2} we introduce the self-normalized residual and the Jarque-Bera type statistics:
\begin{align*}
&\hat{N}_j=\hat{S}_n^{-1/2}(\ep_j(\aes)-\bar{\hat{\ep}}_n),\\
&\mathrm{JB}_n= \frac{1}{6n}\left(\sumj (\hat{N}_j)^3-3\sqrt{h_n}\sumj \p_x a_{j-1}(\aes)\right)^2+\frac{1}{24n}\left(\sumj((\hat{N}_j)^4-3)\right)^2,
\end{align*}
where 
\begin{equation*}
\bar{\hat{\ep}}_n:=\frac{1}{n}\sumj \ep_j(\aes), \quad \hat{S}_n:=\frac{1}{n}\sumj(\ep_j(\aes)-\bar{\hat{\ep}}_n)^2.
\end{equation*}
The following theorem gives the asymptotic behavior of $\mathrm{JB}_n$, which ensures theoretical validity of the Jarque-Bera type test based on $\mathrm{JB}_n$.

\begin{Thm}(\cite[Theorems 3.1 and 4.1]{Mas13-2})
\label{Achi&p}
\begin{enumerate}
\item Under $\mch_0:  {\lam=0}$ and suitable regularity conditions, for any estimator $\aes$ of $\al$ satisfying
\begin{equation}\label{sqn}
\sqrt{n}(\aes-\al_0)=O_p(1),
\end{equation}
we have
\begin{equation*}
\mathrm{JB}_{n}\cil \chi^2(2).
\end{equation*}

\item Under $\mch_1:   {\lam>0}$ and suitable regularity conditions, we have
\begin{equation*}
\mathrm{JB}_{n}\overset{P}\rightarrow \infty,
\end{equation*}
that is, $P(\mathrm{JB}_{n}>K) \to 1$ for any $K>0$.
\end{enumerate}
\end{Thm}


\begin{Rem}
The residual defined by \eqref{hm:ep_def} is of the Euler type with ignoring the drift fluctuation;
under the sampling conditions in Assumption \ref{Sampling} given later, we can ignore the presence of the drift term in construction of residuals.
Indeed, instead of \eqref{hm:ep_def} we could consider
\begin{equation}
\ep_j(\theta)=\ep_{n,j}(\theta):=\frac{\D_{j}X - h_n b_{j-1}(\beta) }{\sqrt{a^{2}_{j-1}(\al)h_n}}.
\nonumber
\end{equation}
Also, we could define $\mathrm{JB}_{n}$ only by the skewness or kurtosis part; this only changes the asymptotic degrees of freedom $2$ in Theorem \ref{Achi&p}-(1) by $1$.
See \cite{Mas11} for the technical details.
This case may require more computation time, while we would then have a stabilized performance under $\mch_0$ compared with the case of \eqref{hm:ep_def}.
\end{Rem}

\begin{Rem}
The results of \cite{Mas13-2} can apply even when the jump component is driven by a compound Poisson process, possibly a much broader class of finite-activity processes. It is therefore expected that we may relax the structural assumption, although the theoretical results in Section \ref{Asymptotic Results} then require a large number of modifications.
\end{Rem}

In the rest of this section, suppose that the null hypothesis $\mch_0$ is true, so that the underlying model is the diffusion process.
Among choices of $\aes$, the Gaussian quasi-maximum likelihood estimator (GQMLE) is one of the most important candidates because it has the asymptotic efficiency in H\'{a}jek-Le Cam sense (cf. \cite{Gob02}).
The GQMLE is defined as any maximizer of the Gaussian quasi-likelihood (GQL)
\begin{equation*}
\mbbh_{n}(\theta) := \sumj \log\left\{\frac{1}{\sqrt{2\pi a^{2}_{j-1}(\al)h_n}}\phi\left(\frac{\D_{j}X - b_{j-1}(\beta)h_n}{\sqrt{a^{2}_{j-1}(\al)h_n}}\right)\right\},
\end{equation*}
where $\phi$ denotes the standard normal density.
This quasi-likelihood is constructed based on the local-Gauss approximation of the transition probability $\mcl(X_{t_{j}}|X_{t_{j-1}})$ by $N(b_{j-1}(\beta)h_n, a^{2}_{j-1}(\al)h_n)$.
It is well known that the asymptotic normality holds true under suitable regularity conditions  \cite{Kes97}: For the GQMLE $\tilde{\theta}_n=(\tilde{\al}_n,\tilde{\beta}_n)$, we have
\begin{equation}
\left( \sqrt{n}(\tilde{\al}_n-\al_0),\, \sqrt{T_n}(\tilde{\beta}_n-\beta_0) \right) \cil N\left(0, \,\diag(I_{1}^{-1}(\tz), I_{2}^{-1}(\tz))\right),
\nonumber
\end{equation}
where
\begin{align*}
&I_{1}(\tz)=\frac{1}{2}\int \left(\frac{\p_\al a^2}{a^2}(x,\al_0)\right)^{\otimes2}\pi_0(dx),\\
&I_{2}(\tz)=\int\left(\frac{\p_\beta b}{a}(x,\beta_0)\right)^{\otimes2}\pi_0(dx),
\end{align*}
both assumed to be positive definite.
Here $\pi_0$ denotes the invariant measure of $X$.

The strategy we will describe in Section \ref{Proposed strategy} is in principle valid even when the drift and diffusion coefficients are nonlinear in the parameters.
However, if the coefficients $a$ and $b$ are highly nonlinear and/or the number of the parameters is large, then the calculation of the GQMLE can be quite time-consuming.
To deal with such a problem, it is effective to separate optimizations of $\al$ and $\beta$ by utilizing the difference of the small-time stochastic orders of the $dt$- and $dw_t$-terms.
To be specific, we introduce the following stepwise version of the GQMLE $\check{\theta}_n:=(\check{\al}_n,\check{\beta}_n)$:
\begin{align*}
&\check{\al}_n\in\argmax_{\al\in\bar{\Theta}_\al} \sumj \log\left\{\frac{1}{\sqrt{2\pi a^{2}_{j-1}(\al)h_n}}\phi\left(\frac{\D_{j}X}{\sqrt{a^{2}_{j-1}(\al)h_n}}\right)\right\},\\
&\check{\beta}_n\in\argmax_{\beta\in\bar{\Theta}_\beta}
\mbbh_{n}(\check{\al}_n,\beta).
\end{align*}
Under some suitable regularity condition, it is shown that the stepwise GQMLE has the same asymptotic distribution as the original GQMLE $\tilde{\theta}_n$ (cf. \cite{UchYos12}).
Hence $\check{\theta}_n$ is asymptotically efficient, and the claims in Theorem \ref{Achi&p} with $\aes$ replaced by $\check{\al}_n$ hold true.
Although in general we have to conduct two optimization for the stepwise estimation scheme, it lessens the number of the parameters to be simultaneously optimized, thus reducing the computational time.

\section{Proposed strategy}\label{Proposed strategy}
In this section, still looking at \eqref{yu:sde}, we propose an iterative jump detection procedure based on the Jarque-Bera type test introduced in the previous section.

Let $q\in(0,1)$ be a small number, which will later serve as the significance level.
Suppose that we are given an estimator $\tes$ of $\theta=(\al,\beta)$ defined to be any element $\tes \in \argmax M_n$ for some contrast function $M_n$ of the from
\begin{equation}
M_n(\theta) := \sumj m_{h_{n}}\left( X_{t_{j-1}},\D_j X;\,\theta \right).
\nonumber
\end{equation}
Denote by $\chi^2_q(2)$ the upper $q$-percentile of the chi-squared distribution with $2$ degrees of freedom.
Then, our procedure is as follows; we implicitly assume that there is no tie among the values $|\D_{1}X|,\dots,|\D_{n}X|$.

\medskip

\begin{itemize}
\item[{\it Step 0.}] Set $k=k_n=0$, and let $\hat{\mcj}_{n}^0:=\emptyset$.

\medskip

\item[{\it Step 1.}]  
Calculate the modified estimator $\tes^k$ defined by
\begin{equation}
\tes^k \in \argmax_{\theta\in\Theta} \sum_{j\notin\hat{\mcj}_n^k} m_{h_{n}}\left( X_{t_{j-1}},\D_j X;\,\theta \right),
\nonumber
\end{equation}
then let
\begin{equation*}
\bar{\hat{\ep}}_n^k:=\frac{1}{n-k}\sum_{j\notin\hat{\mcj}_{n}^k} \ep_j(\aes^k), \qquad \hat{S}_n^k:=\frac{1}{n-k}\sum_{j\notin\hat{\mcj}_{n}^k}(\ep_j(\aes^k)-\bar{\hat{\ep}}_n^k)^2,
\end{equation*}
and (re-)construct the following modified self-normalized residuals $(\hat{N}_j^k)_{j=1}^n$ and Jarque-Bera type statistics $\mathrm{JB}_{n}^k$:
\begin{align}\label{yu:msnr}
\hat{N}_j^k &:=(\hat{S}_n^k)^{-1/2}(\ep_j(\aes^k)-\bar{\hat{\ep}}_n^k), \nn\\
\mathrm{JB}_{n}^k &:= \frac{1}{6(n-k)}\left(\sum_{j\notin\hat{\mcj}_{n}^k} (\hat{N}_j^k)^3-3\sqrt{h_n}\sum_{j\notin\hat{\mcj}_{n}}\p_x a_{j-1}(\aes^k)\right)^2
\\
&{}\qquad
+\frac{1}{24(n-k)}\left(\sum_{j\notin\hat{\mcj}_{n}^k}((\hat{N}_j^k)^4-3)\right)^2.\nn
\end{align}

\medskip

\item[{\it Step 2.}]
If $\mathrm{JB}_{n}^k>\chi^2_q(2)$, then pick out the interval number
\begin{equation}
j(k+1):=\argmax_{j\in\{1,\dots, n\}\setminus\hat{\mcj}_{n}^k} |\D_{j}X|,
\nonumber
\end{equation}
add it to the set $\hat{\mcj}_{n}^k$:
\begin{equation}
\hat{\mcj}_{n}^{k+1} := \hat{\mcj}_{n}^{k} \cup \{j(k+1)\},
\nonumber
\end{equation}
and then return to {\it Step 1}.
If $\mathrm{JB}_{n}^k \le \chi^2_q(2)$, then set an estimated number of jumps to be
\begin{equation}
k^\star=k_n^\star(\omega) := \min\left\{ k\le n;~ \mathrm{JB}_{n}^k \le \chi^2_q(2) \right\}
\nonumber
\end{equation}
and go to {\it Step 3}.

\medskip

\item[{\it Step 3.}]
If $k^\star=0$, regard that there is no jump; otherwise, we regard that each of $\D_{j(1)}X, \dots, \D_{j(k^\star)} X$ contains one jump.
Finally, set $\tes^{k^\star}$ to be an estimator of $\theta$.
\end{itemize}

\medskip

In practice, the above-described method enables us to divide the set of the whole increments $(\D_{j}X)_{j=1}^{n}$ into the following two categories:
\begin{itemize}
\item ``One-jump'' group $(\D_j X)_{j\in\hat{\mcj}_n^{k^{\star}}}=\{ \D_{j(1)}X,\dots,\D_{j(k^\star)}X\}$, and
\item ``No-jump'' group $(\D_j X)_{j\notin\hat{\mcj}_n^{k^{\star}}}=(\D_{j}X)_{j=1}^{n} \setminus \{ \D_{j(1)}X,\dots,\D_{j(k^\star)}X\}$.
\end{itemize}
Automatically entailed just after jump removals are stopped is the estimator $\tes^{k^{\star}}$ of the drift and diffusion parts of $X$,
which is the maximizer of the \textit{modified Gaussian quasi-likelihood} defined by
\begin{equation*}
\theta \mapsto \sum_{j\notin\hat{\mcj}_{n}^{k^\star}}
\log\left\{\frac{1}{\sqrt{2\pi a^{2}_{j-1}(\al)h_n}}\phi\left(\frac{\D_{j}X - b_{j-1}(\beta)h_n}{\sqrt{a^{2}_{j-1}(\al)h_n}}\right)\right\}.
\end{equation*}
As is demonstrated in Section \ref{Asymptotic Results}, our primary setting \eqref{hm:sde} is designed not to require any optimization using a numerical search such as the quasi-Newton method.

We should note that, due to the nature of the testing, there may remain positive probability of spurious detection of jumps no matter how large number of data is.
Nevertheless, as long as the underlying model is correct, the number of removals is much smaller than the total sample size, so that spurious removals are not serious here.

\medskip

\begin{Rem}\label{shift}
In the above-described procedure we simply remove the largest increments at each step, with keeping the positions of the remaining data.
Note that in the construction of the modified estimator $\tes^k$ it is incorrect to use the ``shifted" samples $(Y_{t_j})_{j\notin\hat{\mcj}_n^{k_n}}$ defined by
\begin{equation*}
Y_{t_j}=X_{t_j}-\sum_{i\in\hat{\mcj}_n^{k_n}\cap\{1,\dots,j\}}\D_i X.
\end{equation*}
This is because one-step transition density of the original process $X$ is spatially different from $Y$, so that the estimation result would not suitably reflect the information of data.
\end{Rem}

\begin{Rem}\label{ite}
At $k$-th iteration, it can be regarded that we conduct the Jarque-Bera type test for the trimmed data $(X_{t_{j-1}},\D_j X)_{j\notin\hat{\mcj}_n^k}$. 
Hence the null hypothesis $\mch^k_0$ and alternative hypothesis $\mch^k_1$ of the test are formally written as follows:
\begin{align*}
&\mch^k_0: \sharp\left\{j\in\{1,\dots,n\} \ \middle| \ \D_j N\geq 1 \right\}\leq k,\\
&\mch^k_1: \sharp\left\{j\in\{1,\dots,n\} \ \middle| \ \D_j N\geq 1 \right\}>k,
\end{align*}
where $\sharp A$ denotes the cardinality of any set $A$.
From this formulation, we have the inclusion relation
\begin{equation*}
\mch_0\subset \mch_0^1\subset \mch_0^2 \subset \dots \subset \mch_0^k\subset \cdot\cdot\cdot,
\end{equation*} 
which implicitly suggests that we can extract more than one increments at \textit{Step 2} when seemingly several jumps do exist:
indeed, in view of the expectation of Poisson processes, it seems reasonable to remove at the first rejection of $\mch_0$ not only $|\D_{j(1)}X|$ but the first $O(T_n)$ largest increments, resulting in acceleration of terminating the procedure.
\end{Rem}

\begin{Rem}
In practice, the size of ``last-removed'' increment:
\begin{equation}
r_{n}(k^{\star}):=|\D_{j(k^{\star})}X|
\nonumber
\end{equation}
would be used as a threshold for detecting jumps for future increments.
\end{Rem}

\begin{Rem}
\label{hm:rem_F.esti}
When the jump coefficient is parameterized as $c(x,\gam)$ and a model of the common jump distribution, say $F_J$, of the compound Poisson process $J$ is given, we may consider estimation of $\gam$ and $F_J$ based on the sequence $\{\D_{j(k)}X/c_{j(k)-1}(\gam)\}_{k=1}^{k^\star}$, with supposing that they are i.i.d. random variables with common jump distribution $F_J$; note that number of jumps tends to increase for a larger $T_n$.
This is beyond the scope of this paper, and we leave it as a future study.
\end{Rem}

\section{Asymptotic results}\label{Asymptotic Results}

We now return to the model \eqref{hm:sde}.
As was mentioned in the previous section, we have a choice of an estimator of $\theta$.
As a matter of course, for each estimator $\tes$, we need to study asymptotic behavior of its modified version $\tes^{k*}$.
In this section, we will derive asymptotic results for a numerically tractable least-squares type estimator and the corresponding one-step improved version. 
For simplicity, we write 
\begin{align*}
&\mbba(x)=(a^{(1)}(x),\dots,a^{(p_\al)}(x))^\top, \quad \mbbb(x)=(b^{(1)}(x),\dots,b^{(p_\beta)}(x))^\top.
\end{align*}

\begin{Assumption}[Regularity of coefficients]\label{Ascoef}
The following conditions hold:
\begin{enumerate}
\item $\ds{0< \inf_{x,\al}\mbba(x)^{\top}\al \wedge \inf_x |c(x)|}$ \ and \ $\ds{\sup_{x,\al}\mbba(x)^{\top}\al \vee \sup_x |c(x)| <\infty}$;
\item $\ds{\left|\sqrt{\mbba(x)^\top\al_0}-\sqrt{\mbba(y)^\top\al_0}\right|+\left|\mbbb(x)-\mbbb(y)\right|+\left|c(x)-c(y)\right|\lesssim |x-y|, \quad x,y\in\mbbr}$;
\item There exists a constant $C'\ge 0$ for which $\ds{|\p_x \mbba(x)| + |\p_x^{2} \mbba(x)| \lesssim 1+|x|^{C'}, \quad x\in\mbbr}$.
\end{enumerate}
\end{Assumption}

Here the supremum with respect to $\al$ is taken over the compact set $\bar{\Theta}_\al$.
The basic scenario to construct an estimator of $\theta$ when $X$ had no jumps is as follows:
\begin{itemize}
\item We first estimate the diffusion parameter by the least-squares estimator (LSE):
\begin{align}
\tilde{\al}_n &:= \argmin_\al \sumj \left\{(\D_j X)^2-h_n \mbba_{j-1}^\top\al\right\}^2 \nn\\
&=\frac{1}{h_n} \left(\sumj \mbba_{j-1}\mbba_{j-1}^\top\right)^{-1} \sumj (\D_j X)^2\mbba_{j-1}. \nn\\
\nonumber
\end{align}

\item We then improve the LSE through the scoring with the GQL:
\begin{equation}\label{ose}
\aes:= \tilde{\al}_n-\left(\sumj\frac{\mbba_{j-1}\mbba_{j-1}^\top}{(\mbba_{j-1}^\top\tilde{\al}_n)^2}\right)^{-1}\sumj \left(\frac{1}{\mbba_{j-1}^\top\tilde{\al}_n}-\frac{(\D_j X)^2}{h_n(\mbba_{j-1}^\top \tilde{\al}_n)^2}\right)\mbba_{j-1}.
\end{equation}
\item Finally we estimate the drift parameter by the plug-in LSE:
\begin{align}
\bes&:=\argmin_\beta \sumj \frac{(\D_j X-h_n \mbbb_{j-1}^\top\beta)^2}{\mbba_{j-1}^\top\aes} \nn\\
&=\frac{1}{h_n}\left(\sumj \frac{\mbbb_{j-1}\mbbb_{j-1}^\top}{\mbba_{j-1}^\top\aes}\right)^{-1} \sumj \frac{\D_j X}{\mbba_{j-1}^\top\aes}\mbbb_{j-1}.
\nn
\end{align}
\end{itemize}
It is known that $\tilde{\al}_n$ is not asymptotically efficient, 
while $\bes$ is in case where the underlying process is a diffusion process, which is why we additionally consider the improved version $\aes$ based on the stepwise GQL: 
\begin{equation*}
\mbbh_{1,n}(\al):=-\frac{1}{2}\sumj\left\{\log\left( 2\pi h_n \mbba_{j-1}^\top\al \right)+\frac{(\D_j X)^2}{h_n\mbba_{j-1}^\top\al}\right\}.
\end{equation*}
Then $\aes$ is asymptotic efficient under appropriate regularity conditions.
The form of the second term in the right-hand side in \eqref{ose} comes from the quasi-score associated with $\mbbh_{1,n}(\al)$ and 
the expression of the Fisher information matrix corresponding to $\al$, where the latter equals the upper left part of $\Sig_0$ in Theorem \ref{osan}.

Now, in the presence of jumps, in view of Section \ref{Proposed strategy} we introduce the modified estimators
\begin{align}
&\tilde{\al}_n^{k_n}=\frac{1}{h_n} \left(\sum_{j\notin\hat{\mcj}_n^{k_n}} \mbba_{j-1}\mbba_{j-1}^\top\right)^{-1} \sum_{j\notin\hat{\mcj}_n^{k_n}} (\D_j X)^2\mbba_{j-1},
\nn\\
&\bes^{k_n}=\frac{1}{h_n}\left(\sum_{j\notin\hat{\mcj}_n^{k_n}} \frac{\mbbb_{j-1}\mbbb_{j-1}^\top}{\mbba_{j-1}^\top\aes^{k_n}}\right)^{-1} \sum_{j\notin\hat{\mcj}_n^{k_n}} \frac{\D_j X}{\mbba_{j-1}^\top\aes^{k_n}}\mbbb_{j-1},
\nn
\end{align}
where $\aes^{k_n}$ is the improved estimator defined by
\begin{equation}
\aes^{k_n}=\tilde{\al}_n^{k_n}-\left(\sum_{j\notin\hat{\mcj}_n^{k_n}}\frac{\mbba_{j-1}\mbba_{j-1}^\top}{(\mbba_{j-1}^\top\tilde{\al}_n^{k_n})^2}\right)^{-1}\sum_{j\notin\hat{\mcj}_n^{k_n}} \left(\frac{1}{\mbba_{j-1}^\top\tilde{\al}_n^{k_n}}-\frac{(\D_j X)^2}{h_n(\mbba_{j-1}^\top \tilde{\al}_n^{k_n})^2}\right)\mbba_{j-1}.
\label{hm:ose-a}
\end{equation}
The inverse matrices appearing in the above definitions asymptotically exist under the forthcoming conditions, hence implicitly assumed here for brevity.
What is important from these expressions is that we can calculate the modified estimators $\tilde{\al}_n^{k_n}$, $\bes^{k_n}$, and $\aes^{k_n}$ simply by removing the indices in $\hat{\mcj}_n^{k_n}$ in computing the sums without repetitive numerical optimizations, thus reducing the computational time to a large extent.
Further, it should also be noted that we may proceed only with $\tilde{\al}_n^{k_n}$ without the improved version $\hat{\al}_n^{k_n}$, if the asymptotically efficient estimator is not the first thing to have and quick-to-compute estimator is more needed.

\medskip

To state our main result, we introduce further assumptions below.

\begin{Assumption}[Stability]\label{Moments}$\ $
\begin{enumerate}
\item There exists a unique invariant probability measure $\pi_0$, and for any function $f\in L_1(\pi_0)$, we have
\begin{equation*}
\frac{1}{T}\int_0^T f(X_t)dt\cip \int_\mbbr f(x)\pi_0(dx), \quad \mathrm{as} \ T\to\infty.
\end{equation*} 
\item $\ds{\sup_{t\in\mbbr^+} E[|X_t|^q]<\infty}$ for any $q>0$.
\end{enumerate}
\end{Assumption}

\medskip

\begin{Assumption}[Sampling design]\label{Sampling}
There exist positive constants $\kappa',\kappa \in (1/2,1)$ such that
\begin{equation}
n^{-\kappa'} \lesssim h_n \lesssim n^{-\kappa}.
\nonumber
\end{equation}
\end{Assumption}

Recall that the driving noise $J$ can be expressed as 
\begin{equation*}
J_t=\sum_{i=1}^{N_t} \xi_i,
\end{equation*}
by a Poisson process $N$ and i.i.d random variables $(\xi_i)$ being independent of $N$.

\begin{Assumption}[Jump size]\label{Asjsize}$\ $
\begin{enumerate}
\item $E[|\xi_1|^q]<\infty$ for any $q>0$.
%
%
\item In addition to Assumption \ref{Sampling},
\begin{equation}
\limsup_{x\downarrow 0} x^{-s}P\left( |\xi_1| \le x \right) < \infty,
\label{hm:add-3}
\end{equation}
for some constant $s$ satisfying
\begin{equation}
s > \frac{4(1-\kappa)}{2\kappa -1}.
\nonumber
\end{equation}
\end{enumerate}
\end{Assumption}

\medskip

Here are some technical remarks on each assumption.
Assumption \ref{Ascoef} ensures the existence of a {\cadlag} solution of \eqref{hm:sde}, and its Markovian property (cf. \cite[chapter 6]{App09}).
Assumption \ref{Moments} is essential to derive our theoretical results. In our Markovian framework, it suffices for Assumption \ref{Moments}-(1) to have
\begin{equation}
\left\| P_t(x,\cdot) -\pi(\cdot) \right\|_{TV} \to 0, \quad t\to\infty, \quad x\in\mbbr,
\label{hm:ergodicity}
\end{equation}
for some probability measure $\pi_0$, where $\| \mathfrak{m}(\cdot) \|_{TV}$ denotes the total variation norm of a signed measure $\mathfrak{m}$ and $\{P_t(x,dy)\}$ does the family of transition probability of $X$; then, $\pi_0$ is the unique invariant measure of $X$, and Assumption \ref{Moments}-(1) holds for any $f\in L_1(\pi_0)$ and any initial distribution $\mcl(X_0)$, see \cite{Bha82} for details.
Further, \eqref{hm:ergodicity} with Assumption \ref{Moments}-(2) implies that
\begin{equation}
\int_\mbbr|x|^q\pi_0(dx)<\infty
\nonumber
\end{equation}
for any $q>0$; this can be seen in a standard manner using Fatou's lemma and the monotone convergence theorem through a smooth truncation of the mapping $x\mapsto|x|^q$ into a compact set.
We refer to \cite{Kul09}, \cite{Mas08}, and \cite{Mas13-1} for an easy-to-check sufficient condition for \eqref{hm:ergodicity} and Assumption \ref{Moments}-(2).

Assumptions \ref{Sampling} and \ref{Asjsize} describe a tradeoff between sampling frequency and probability of small jump size (quicker decay of $h_n$ allows for more frequent small jumps of $J$).
We have formulated them with giving preference to simplicity over complexity. See Section \ref{hm:sec_pre-rems} for some technical consequences which we will really require in the proofs.

\begin{Rem}
We are focusing on estimation of both drift and diffusion coefficients under the ergodicity.
Nevertheless, we may consistently estimate the diffusion coefficient even when the terminal sampling time is fixed, such as $T_n\equiv 1$, without ergodicity; see \cite{GenJac93}, and also \cite{InaYos18} as well as the references therein.
Since \cite{Mas13-2} can handle the non-ergodic case as well, it is expected that our estimation strategy in Section \ref{Proposed strategy} would remain in place and the theoretical results in this section would have trivial non-ergodic counterparts, to be valid under much weaker assumptions; in particular, we would only require \eqref{hm:add-3} for some $s>0$.
\end{Rem}

\medskip

To investigate the asymptotic property of our estimators, we introduce the unobserved continuous part of $X$ defined by
\begin{align*} 
\Xc_t=X_t-X_0-\int_0^t c(X_{s-})dJ_s=\int_0^t a(X_s,\al_0)dw_t+\int_0^t b(X_s,\beta_0)dt.
\end{align*}
Let $(\check{\al}_n)$ be any random sequence such that
\begin{equation}\label{ines}
\sqrt{n}(\check{\al}_n-\al_0)=O_p(1).
\end{equation}
As in \eqref{hm:ose-a}, we define the random sequence $\aes^{\mathrm{cont}}$ by
\begin{equation*}
\aes^{\mathrm{cont}}=\check{\al}_n-\left(\sumj\frac{\mbba_{j-1}\mbba_{j-1}^\top}{(\mbba_{j-1}^\top\check{\al}_n)^2}\right)^{-1}\sumj\left(\frac{1}{\mbba_{j-1}^\top\check{\al}_n}-\frac{(\D_j \Xc)^2}{h_n(\mbba_{j-1}^\top \check{\al}_n)^2}\right)\mbba_{j-1}.
\end{equation*}
Correspondingly, we also define
\begin{equation*}
\bes^{\mathrm{cont}}:=\frac{1}{h_n}\left(\sumj\frac{\mbbb_{j-1}\mbbb_{j-1}^\top}{\mbba_{j-1}^\top\aes^{\mathrm{cont}}}\right)^{-1} \sumj \frac{\D_j \Xc}{\mbba_{j-1}^\top\aes^{\mathrm{cont}}}\mbbb_{j-1}.
\end{equation*}
As is expected, $(\aes^{\mathrm{cont}},\bes^{\mathrm{cont}})$ serves as a good estimator if it could be computed:

\begin{Thm}\label{osan}
Suppose that Assumptions \ref{Ascoef} to \ref{Sampling}, and Assumption \ref{Asjsize}-(1) hold, and that both $\int \mathbb{A}(x)^{\otimes 2}\pi_0(dx)$ and $\int \mathbb{B}(x)^{\otimes 2}\pi_0(dx)$ are positive definite.
Then we have
\begin{align*}
\left(\sqrt{n}(\aes^{\mathrm{cont}}-\al_0), \sqrt{T_n}(\bes^{\mathrm{cont}}-\beta_0)\right)\cil N\left(0, \Sig_0\right),
\end{align*}
where 
\begin{equation*}
\Sig_0:=\begin{pmatrix}
\displaystyle2\left\{\int\left(\frac{\mbba(x)}{(\mbba(x))^\top\al_0}\right)^{\otimes2}\pi_0(dx)\right\}^{-1} & \displaystyle O\\
\displaystyle O &\displaystyle \left\{\int\frac{\mbbb^{\otimes2}(x)}{\mbba(x)^\top\al_0}\pi_0(dx)\right\}^{-1}
\end{pmatrix}.
\end{equation*}
\end{Thm}

\begin{Rem}
\label{hm:rem_asymp.eff}
The asymptotic covariance matrix of $\hat{\beta}^{\mathrm{cont}}_n$ is formally the efficient one, see \cite[Theorem 2.2]{KohNuaTra17}.
Moreover, that of $\hat{\al}^{\mathrm{cont}}_n$ is the same as the estimator in \cite{ShiYos06} and \cite{OgiYos11} based on a jump-detection filter.
\end{Rem}

The next theorem states that, asymptotically, on the set $ \left\{ \mathrm{JB}_n^{k_n}\leq\chi^2_q(2)\right\}$, the number of jumps is less than $k_n$, and thus the modified LSE type diffusion estimator $\tilde{\al}_n^{k_n}$ consists of (true) "no-jump" group and has the $\sqrt{n}$-consistency.

\begin{Thm}\label{Consistency}
Suppose that Assumptions \ref{Ascoef} to \ref{Asjsize} hold, and that both $\int \mathbb{A}(x)^{\otimes 2}\pi_0(dx)$ and $\int \mathbb{B}(x)^{\otimes 2}\pi_0(dx)$ are positive definite.
Then, for any $\ep>0$, we can find a sufficiently large $M>0$ and $N\in\mbbn$ such that
\begin{equation}
\sup_{n\ge N}P\left(\left\{|\sqrt{n}(\tilde{\al}_n^{k_n}-\al_0)|>M\right\}\cap\left\{\mathrm{JB}_n^{k_n}\leq\chi^2_q(2)\right\}\right)<\ep.
\nn
\end{equation}
\label{thm_consis1}
\end{Thm}


By re-defining $(\tilde{\al}_n^{k_n})$ as 
\begin{equation}
\tilde{\al}_n^{k_n}=\begin{cases}\tilde{\al}_n^{k_n} & \quad \mathrm{on} \ \ \left\{\mathrm{JB}_n^{k_n}\leq\chi^2_q(2)\right\}, \\
\al_0 & \quad \mathrm{on} \ \ \left\{\mathrm{JB}_n^{k_n}>\chi^2_q(2)\right\},
\end{cases}
\end{equation}
$(\tilde{\al}_n^{k_n})$ enjoys the property \eqref{ines}: $\sqrt{n}(\tilde{\al}_n^{k_n}-\al_0)=O_p(1)$ from Theorem \ref{thm_consis1}, so that by Theorem \ref{osan}, we have
\begin{align*}
\left(\sqrt{n}( \aes^{k_n, \mathrm{cont}}-\al_0), \sqrt{T_n}(\bes^{k_n, \mathrm{cont}}-\beta_0)\right)\cil N\left(0, \Sig_0\right),
\end{align*}
where
\begin{align*}
&\aes^{k_n, \mathrm{cont}}:=\tilde{\al}_n^{k_n}-\left(\sumj\frac{\mbba_{j-1}\mbba_{j-1}^\top}{(\mbba_{j-1}^\top\tilde{\al}_n^{k_n})^2}\right)^{-1}\sumj\left(\frac{1}{\mbba_{j-1}^\top\tilde{\al}_n^{k_n}}-\frac{(\D_j \Xc)^2}{h_n(\mbba_{j-1}^\top \tilde{\al}_n^{k_n})^2}\right)\mbba_{j-1},\\
&\bes^{k_n, \mathrm{cont}}:=\frac{1}{h_n}\left(\sumj\frac{\mbbb_{j-1}\mbbb_{j-1}^\top}{\mbba_{j-1}^\top\aes^{k_n}}\right)^{-1} \sumj \frac{\D_j \Xc}{\mbba_{j-1}^\top\aes^{k_n}}\mbbb_{j-1}.
\end{align*}

Recall that we finish our procedure once we have $\mathrm{JB}_n^{k_n}\leq\chi^2_q(2)$.
The following theorem is the main claim of this section.

\begin{Thm}\label{Ae}
Suppose that Assumptions \ref{Ascoef} to \ref{Asjsize} hold and that $\Sig_0$ in Theorem \ref{osan} is positive definite.
Then, for any $\ep>0$ and $q\in(0,1)$, we have
\begin{align}\label{ae}
&P\left(\left\{\left|\sqrt{n}(\aes^{k_n}-\aes^{k_n,\mathrm{cont}})\right|\vee\left|\sqrt{T_n}(\bes^{k_n}-\bes^{k_n,\mathrm{cont}})\right|>\ep\right\}\cap \left\{ \mathrm{JB}_n^{k_n}\leq\chi^2_q(2)\right\}\right)\to0.
\end{align}
\end{Thm}

\begin{Rem}
Since each phase of our method is conducted on the null hypothesis, we do not identify the true value in the re-defined $\tilde{\al}_n^{k_n}$ in practice.
\end{Rem}

\begin{Rem}
We should note that the number of jump removals is automatically determined by the iterative Jarque-Bera type test, and thus there is no need to choose $(k_n)$ in practice.
\end{Rem}

\section{Numerical experiments}\label{Numerical experiments}

\begin{table}[t]
\caption{The performance of our estimators is given in case (i).
The mean is given with the standard deviation in parenthesis. In this table, $k_n^\star$ denotes the number of jumps.}
\label{res3-1}
\begin{center}
\begin{tabular}{cccccccccc}
\hline
&&&&&&&&&\\[-3.5mm]
$T_n$ & $n$ & $h_n$ & $k_n^\star$ & \multicolumn{6}{c}{(i)$\text{Gamma distribution}$}  \\ 
&&&& $\aes^{0}$ & $\bes^{0}$ &$\aes^{k_n}$  &$\bes^{k_n}$  & $\aes^{k_n^\star}$&$\bes^{k_n^\star}$ \\ \hline
28.8& 1000 &
0.03&15&18.80&0.62&3.38&0.99&3.38&1.00 \\
&&&&(4.31)&(0.13)&(0.20)&(0.09)&(0.20)&(0.09)\\ 
62.1&10000&0.006&30&17.7&0.63&3.07&1.00&3.08&1.00\\  
&&&&(2.91)&(0.09)&(0.05)&(0.06)&(0.04)&(0.06)\\ \hline
\end{tabular}
\end{center}
\end{table}

\begin{table}[t]
\caption{The performance of our estimators is given in case (ii).
The mean is given with the standard deviation in parenthesis.
In this table, $k_n^\star$ denotes the number of jumps.}
\label{res3-2}
\begin{center}
\begin{tabular}{cccccccccc}
\hline
&&&&&&&&&\\[-3.5mm]
$T_n$ & $n$ & $h_n$ & $k_n^\star$ & \multicolumn{6}{c}{(ii)$\text{Bilateral inverse Gaussian distribution}$}  \\ 
&&&& $\aes^{0}$ & $\bes^{0}$ &$\aes^{k_n}$  &$\bes^{k_n}$  & $\aes^{k_n^\star}$&$\bes^{k_n^\star}$ \\ \hline
28.8& 1000 &
0.03&15&10.83&0.82&3.19&0.99&3.15&1.00 \\
&&&&(3.70)&(0.22)&(0.17)&(0.14)&(0.16)&(0.14)\\ 
62.1&10000&0.006&30&10.22&0.82&3.04&1.01&3.04&1.01\\  
&&&&(2.46)&(0.15)&(0.06)&(0.09)&(0.05)&(0.09)\\ \hline
\end{tabular}
\end{center}
\end{table}

In this section, we conduct Monte Carlo simulation in order to see the performance of our method.
First we consider the following statistical model:
\begin{equation}\label{yu:simmodel}
dX_t=\sqrt{\frac{\al}{1+\sin^2 X_t}}dw_t-\beta X_tdt+dJ_t\quad X_0=0,
\end{equation}
with the true value $\tz:=(\al_0,\beta_0)=(3,1)$.
As the jump size distributions, we set:
\begin{itemize}
\item[(i)] Gamma distribution $\Gam(4,1)$ (one-sided positive jumps);
\item[(ii)] Bilateral inverse Gaussian distribution $bIG(2,1,4,1)$ (two-sided jumps).
\end{itemize}
The bilateral inverse Gaussian random variable $X\sim bIG(\del_1,\gam_1,\del_2,\gam_2)$ is defined as the difference of two independent inverse Gaussian random variable $X_1\sim IG(\del_1,\gam_1)$ and $X_2\sim IG(\del_2,\gam_2)$.
In the trials, we set the significance level $q=10^{-3}$, and the number of jumps fixed just for purposes of comparison.

Based on independently simulated 1000 sample path, the mean and standard deviation of our estimator $(\aes^{k_n},\bes^{k_n})$ are tabulated in Table \ref{res3-1} and Table \ref{res3-2} with the estimators $(\aes^0,\bes^0)$ and $(\aes^{k_n^\star},\bes^{k_n^\star})$.
The first estimator $(\aes^0,\bes^0)$ is constructed by the whole data, and the latter estimator $(\aes^{k_n^\star},\bes^{k_n^\star})$ is constructed by the true no-jump group. 

From these tables, the following items are indicated:
\begin{itemize}
\item In both case, the modified estimators get closer and closer to the true value as jump removals proceed.
\item Since the performances of $(\aes^{k_n},\bes^{k_n})$ and $(\aes^{k_n^\star},\bes^{k_n^\star})$ are almost the same, the jump detection by our method works well.
\item Concerning the drift estimator, the degree of improvement is not large for (ii) relative to (i). It may be due to the two-sided jump structure of $bIG(2,1,4,1)$; thus the amount of improvement is generally expected to be much more significant when the jump distribution is skewed. 
\item In the estimator $(\aes^0,\bes^0)$, the performance of $\aes^0$ is worse than $\bes^0$.
This is because the diffusion estimator is based on the square of the increments $(\D_j X)_j$, thus being heavily affected by jumps.
\item Overall, the diffusion parameter are overestimated even by $\aes^{k_n^\star}$.
Taking into consideration the fact that the mean-reverting point of $X$ is 0, the magnitude of the increment should be larger after one jump occurs.
Thus, although jumps are correctly picked, such overestimation can be seen.
\end{itemize}

\bigskip

\noindent
\textbf{Acknowledgement.}
This work was supported by JST, CREST Grant Number JPMJCR14D7, Japan.

\bigskip
\appendix

\section{Proofs of the result in Section \ref{Asymptotic Results}}
\label{hm:sec_proofs}

Throughout this section, Assumptions \ref{Ascoef} to \ref{Asjsize} are in force. 

\subsection{Technical remarks}
\label{hm:sec_pre-rems}

We summarize a few consequences of Assumptions \ref{Sampling} and \ref{Asjsize}.
All of them will be used later on.

\begin{itemize}
\item Under Assumption \ref{Sampling}, $T_n\to\infty$ and there exists a positive constant $\del=\del(\kappa,\kappa')\in(0,1)$ such that
\begin{equation}
\frac{\log n}{T_n} \vee \left(n^{1+\del}h_n^{2+\del}(\log n)^2\right) \to 0.
\label{hm:add-4}
\end{equation}
This entails $nh_n^2\log n=o\left(T_n^{-\del}/\log n\right)=o(1)$, stronger than the so-called ``rapidly increasing design'': $nh_n\to\infty$ and $nh_n^2\to0$, which is one of standard conditions in the literature of statistical inference for ergodic processes based on high-frequency data.
Assumption \ref{Sampling} will be required for handling the extreme value of the solution process $X$, and for asymptotically allowing the number of jump-removal operations to exceed the expected number of jump times.

\item 
Introduce the positive sequence
\begin{equation}
a_n = a_n(\eta):=T_n^{\eta}
\label{hm:an_def}
\end{equation}
for a constant $\eta>0$.
We see that, with a sufficiently small $\eta$, the following statements follow from Assumptions \ref{Sampling} and \ref{Asjsize}:
\begin{align}
& a_n^3\sqrt{h_n\log n} \to 0, \label{hm:conseq-1}\\
& T_n P\left( |\xi_1|\leq M \sqrt{n}h_n \right) \to 0, \label{hm:conseq-3} \\
& \max_{1\leq j \leq \lf T_n\rf} |\xi_j| = O_p(a_n). \label{hm:conseq-2}
\end{align}
Under Assumption \ref{Sampling}, we can pick an $\eta>0$ small enough to ensure \eqref{hm:conseq-1}; in the sequel, we fix this $\eta$.
Also, we note
\begin{align}
& T_n P\left( |\xi_1|\leq M \sqrt{n}h_n \right) \lesssim n^{-s\kappa/2 + (1-\kappa)(1+s/2)}\to0,
\nonumber
\end{align}
from Assumption \ref{Asjsize}-(2).
As for \eqref{hm:conseq-2}, observe that for any $\ep>0$
\begin{equation*}
P\left(a_n^{-1}\max_{1\leq j \leq \lf T_n\rf} |\xi_j| > \ep\right)=1-\left\{1-P(|\xi_1|> a_n\ep)\right\}^{\lf T_n\rf}.
\end{equation*}
The right-hand side tends to $0$ if the upper bound (due to Markov's inequality) of
\begin{equation}
T_n P(|\xi_1|> a_n\ep) \lesssim n^{1-\kappa-q\eta(1-\kappa')}
\nonumber
\end{equation}
tends to $0$. This holds true under Assumption \ref{Asjsize}-(1) by taking a large constant $q$.

\end{itemize}


\medskip

For abbreviation, we will use the following notations:
\begin{itemize}
\item For any matrix $S=\{S_{kl}\}$, we denote by $|S|:=(\sum_{k,l}S_{kl}^2)^{1/2}$ its Frobenius norm.
\item $I_p$ represents the $p$-dimensional identity matrix.
\item $R(x)$ denotes a differentiable matrix-valued function on $\mbbr$ for which there exists a constant $C>0$ such that $|R(x)|+|\p_x R(x)|\lesssim (1+|x|)^C$, $x\in\mbbr$.
\item We often write $a(x,\al)$ and $b(x,\beta)$ instead of $\sqrt{(\mbba(x))^\top\al}$ and $(\mbbb(x))^\top\beta$.
\item $E^{j-1}[\cdot]$ denotes the conditional expectation with respect to $\mcf_{t_{j-1}}$.
\item We often omit the true value $\tz$, for instance, $a_s$ and $a_{j}$ denote $a(X_s,\al_0)$ and $a(X_{t_j},\al_0)$, respectively.
\end{itemize}

\subsection{Proof of Theorem \ref{osan}}

Let us recall that
\begin{align*} 
\Xc_t:=X_t 
- X_0
-\int_0^t c(X_{s-})dJ_s=\int_0^t a(X_s,\al_0)dw_t+\int_0^t b(X_s,\beta_0)dt.
\end{align*}
First we prove a preliminary lemma.

\begin{Lem}\label{ana}
The $(p_\al+p_\beta)$-dimensional random sequence
\begin{align*}
\left(\frac{1}{\sqrt{n}}\sumj\left(\frac{1}{\mbba_{j-1}^\top\al_0}-\frac{(\D_j \Xc)^2}{h_n(\mbba_{j-1}^\top \al_0)^2}\right)\mbba_{j-1}, \,
\frac{1}{\sqrt{T_n}}\sumj \frac{\D_j \Xc-h_n\mbbb_{j-1}^\top\beta_0}{\mbba_{j-1}^\top\al_0}\mbbb_{j-1}\right)
\end{align*}
weakly converges to the centered normal distribution with covariance matrix
\begin{equation*}
\begin{pmatrix}\displaystyle2\int \left(\frac{\mbba(x)}{(\mbba(x))^\top\al_0}\right)^{\otimes2}\pi_0(dx) &\displaystyle O \\\displaystyle O &\displaystyle \int \frac{\mbbb^{\otimes2}(x)}{\mbba(x)^\top\al_0}\pi_0(dx)\end{pmatrix}.
\end{equation*}
\end{Lem}
\begin{proof}
By the Cram\'{e}r-Wold device, it is enough to show the case where $p_\al=p_\beta=1$.
From the martingale central limit theorem, the desired result follows if we show
\begin{align}
&\label{ford}\frac{1}{\sqrt{n}}\sumj E^{j-1}\left[\left(\frac{1}{a_{j-1}^2}-\frac{(\D_j \Xc)^2}{h_na_{j-1}^4}\right)\mbba_{j-1}\right]\cip0,\\
&\label{sord}\frac{1}{n}\sumj E^{j-1}\left[\left\{\left(\frac{1}{a_{j-1}^2}-\frac{(\D_j \Xc)^2}{h_na_{j-1}^4}\right)\mbba_{j-1}\right\}^2\right]\cip 2\int \left(\frac{\mbba(x)}{a^2(x,\al_0)}\right)^2\pi_0(dx),\\
&\label{foord}\frac{1}{n^2}\sumj E^{j-1}\left[\left\{\left(\frac{1}{a_{j-1}^2}-\frac{(\D_j \Xc)^2}{h_na_{j-1}^4}\right)\mbba_{j-1}\right\}^4\right]\cip0,\\
&\label{bford}\frac{1}{\sqrt{T_n}}\sumj E^{j-1}\left[\frac{\D_j \Xc-h_nb_{j-1}}{a^2_{j-1}}\mbbb_{j-1}\right]\cip0,\\
&\label{bsord}\frac{1}{T_n}\sumj E^{j-1} \left[\left(\frac{\D_j \Xc-h_nb_{j-1}}{a^2_{j-1}}\mbbb_{j-1}\right)^2\right]\cip \int \frac{\mbbb^2(x)}{a^2(x,\al_0)}\pi_0(dx),\\
&\label{bfoord}\frac{1}{(T_n)^2}\sumj E^{j-1} \left[\left(\frac{\D_j \Xc-h_nb_{j-1}}{a^2_{j-1}}\mbbb_{j-1}\right)^4\right]\cip0\\
&\label{cross}\frac{1}{n\sqrt{h_n}}\sumj E^{j-1} \left[\left(\frac{1}{a_{j-1}^2}-\frac{(\D_j \Xc)^2}{h_na_{j-1}^4}\right)\frac{\D_j \Xc-h_nb_{j-1}}{a^2_{j-1}}\mbba_{j-1}\mbbb_{j-1}\right]\cip0.
\end{align}
By using the martingale property of the stochastic integral, Jensen's inequality, the Lipschitz continuity of $b$, and \cite[Lemma 4.5]{Mas13-1}, we have
\begin{align}\label{fice}
E^{j-1}[\D_j \Xc]=h_n b_{j-1}+\intj E^{j-1}[b_s-b_{j-1}]ds=h_nb_{j-1}+h_n^{\frac{3}{2}}R_{j-1}.
\end{align}
It\^{o}'s formula and Fubini's theorem for conditional expectation yield that
\begin{align*}
&E^{j-1}[(\D_j \Xc)^2]\\
&=E^{j-1}\left[2\intj (\Xc_s-\Xc_{j-1})d\Xc_s+\intj (a_s^2-a^2_{j-1})ds+a^2_{j-1}h_n\right]\\
&=a^2_{j-1}h_n+2\intj \left(\int_{t_{j-1}}^sE^{j-1}\left[b_ub_s\right]du\right)ds+\intj E^{j-1}[a_s^2-a^2_{j-1}]ds.
\end{align*}
Again making use of the Lipschitz continuity of $b(x,\beta_0)$ and \cite[Lemma 4.5]{Mas13-1}, we get
\begin{align*}
&\left|\intj \left(\int_{t_{j-1}}^sE^{j-1}\left[b_ub_s\right]du\right)ds\right|\\
&\lesssim \intj \left(\int_{t_{j-1}}^sE^{j-1}\left[1+|X_u|+|X_s|+|X_u||X_s|\right]ds\right)ds\\
&\lesssim h_n^2(1+|X_{j-1}|^2).
\end{align*}
Since $\p_x a^2(x,\al)$ and $\p_x^2 a^2(x,\al)$ are of at most polynomial growth with respect to $x$ uniformly in $\al$, we can similarly deduce that
\begin{align*}
&\left|\intj E^{j-1}[a_s^2-a^2_{j-1}]ds\right|\\
&\lesssim \bigg|\intj E^{j-1}\bigg[\p_x a^2_{j-1}(X_s-X_{j-1}) \nn\\
&{}\qquad +\frac{1}{2}\int_0^1\int_0^1 \p_x^2 a^2(X_{j-1}+uv(X_s-X_{j-1}),\al_0)dudv(X_s-X_{j-1})^2\bigg]ds\bigg| \\
&\lesssim \left|\intj  \left(\int_{t_{j-1}}^sE^{j-1}[b_u]du\right)ds\p_x a^2_{j-1}\right| \nn\\
&{}\qquad + \intj E^{j-1}\left[\left(1+|X_{j-1}|^C+|X_s-X_{j-1}|^C\right)(X_s-X_{j-1})^2\right]ds\\
&\lesssim h_n^2 R_{j-1}.
\end{align*}
Hence
\begin{equation}\label{sce}
E^{j-1}[(\D_j \Xc)^2]=h_na^2_{j-1}+h_n^2R_{j-1}.
\end{equation}
For any $q\geq 2$, Burkholder's inequality for conditional expectation gives
\begin{equation}\label{ece}
E^{j-1}\left[ |\D_j \Xc|^q \right]\lesssim h_n^{\frac{q}{2}} R_{j-1}.
\end{equation}
Repeatedly using It\^{o}'s formula and \eqref{ece}, we have
\begin{align}\label{fce}
&E^{j-1}[(\D_j \Xc)^4]\nonumber\\
&=E^{j-1}\left[4\intj(\Xc_s-\Xc_{j-1})^3d\Xc_s+6\intj (\Xc_s-\Xc_{j-1})^2a_s^2ds\right]\nonumber\\
&=6E^{j-1}\Bigg[\intj \left\{2\int_{t_{j-1}}^s (\Xc_u-\Xc_{j-1})d\Xc_u+\int_{t_{j-1}}^s (a_u^2-a^2_{j-1})du+(s-t_{j-1})a^2_{j-1}\right\} dsa^2_{j-1}\nonumber\\
&\quad \quad+\intj \left\{2\int_{t_{j-1}}^s (\Xc_u-\Xc_{j-1})d\Xc_u +\int_{t_{j-1}}^s a_u^2du
\right\} (a_s^2-a_{j-1}^2)ds\Bigg] \nn\\
&{}\qquad +h_n^{\frac{5}{2}}R_{j-1}\nonumber\\
&=3h_n^2a_{j-1}^4+h_n^{\frac{5}{2}}R_{j-1}.
\end{align}
In particular, it follows from \eqref{sce} and \eqref{fce} that
\begin{align}\label{hm:add-1}
E^{j-1}\left[\left\{\left(\frac{1}{a_{j-1}^2}-\frac{(\D_j \Xc)^2}{h_na_{j-1}^4}\right)\mbba_{j-1}\right\}^2\right]
=\frac{2}{a_{j-1}^4} \mbba_{j-1}^{\otimes 2}+h_n^{\frac{1}{2}}R_{j-1}.
\end{align}
Now, the convergences \eqref{ford}-\eqref{cross} follow from the expressions \eqref{fice}-\eqref{ece} and \eqref{hm:add-1} together with the ergodic theorem (Assumption \ref{Moments}-(1)).
Thus we obtain the desired result. 
\end{proof}

\medskip


Applying Taylor's expansion, we have, with the obvious notation,
\begin{align*}
\frac{1}{n}\sumj\frac{\mbba_{j-1}\mbba_{j-1}^\top}{(\mbba_{j-1}^\top\check{\al}_n)^2}=\frac{1}{n}\sumj\frac{\mbba_{j-1}\mbba_{j-1}^\top}{(\mbba_{j-1}^\top\al_0)^2}+\left\{\frac{1}{n}\int_0^1\sumj\p_\al\left(\frac{\mbba_{j-1}\mbba_{j-1}^\top}{\left[\mbba_{j-1}^\top(\al_0+u(\check{\al}_n-\al_0))\right]^2}\right)du\right\}[\check{\al}_n-\al_0].
\end{align*}
The ergodic theorem implies that
\begin{equation}
\frac{1}{n}\sumj\frac{\mbba_{j-1}\mbba_{j-1}^\top}{(\mbba_{j-1}^\top\al_0)^2} \cip \int \left(\frac{\mbba(x)}{\mbba(x)^\top\al_0}\right)^{\otimes2}\pi_0(dx),
\nonumber
\end{equation}
the limit being positive definite.
From Assumption \ref{Ascoef} and the condition $\sqrt{n}(\check{\al}_n-\al_0)=O_p(1)$, the second term of the right-hand-side is $o_p(1)$.
We also have
\begin{align*}
&\frac{1}{\sqrt{n}}\sumj\left(\frac{1}{\mbba_{j-1}^\top\check{\al}_n}-\frac{(\D_j \Xc)^2}{h_n(\mbba_{j-1}^\top \check{\al}_n)^2}\right)\mbba_{j-1}\\
&=\frac{1}{\sqrt{n}}\sumj\left(\frac{1}{\mbba_{j-1}^\top\al_0}-\frac{(\D_j \Xc)^2}{h_n(\mbba_{j-1}^\top \al_0)^2}\right)\mbba_{j-1}\\
&+\left\{\frac{1}{n}\sumj\left(-\frac{1}{(\mbba_{j-1}^\top\al_0)^2}+2\frac{(\D_j \Xc)^2}{h_n(\mbba_{j-1}^\top \al_0)^3}\right)\mbba_{j-1}\mbba_{j-1}^\top\right\}[\sqrt{n}(\check{\al}_n-\al_0)]+o_p(1).
\end{align*}
By \eqref{sce}, \eqref{fce}, and \cite[Lemma 9]{GenJac93}, it follows that
\begin{equation*}
\frac{1}{n}\sumj\left(-\frac{1}{(\mbba_{j-1}^\top\al_0)^2}+2\frac{(\D_j \Xc)^2}{h_n(\mbba_{j-1}^\top \al_0)^3}\right)\mbba_{j-1}\mbba_{j-1}^\top\cip\int \left(\frac{\mbba(x)}{(\mbba(x))^\top\al_0}\right)^{\otimes2}\pi_0(dx).
\end{equation*}
Hence we obtain
\begin{align*}
\sqrt{n}(\hat{\al}^{\mathrm{cont}}_n-\al_0)=-\left\{\int \left(\frac{\mbba(x)}{(\mbba(x))^\top\al_0}\right)^{\otimes2}\pi_0(dx)\right\}^{-1}\frac{1}{\sqrt{n}}\sumj\left(\frac{1}{\mbba_{j-1}^\top\al_0}-\frac{(\D_j \Xc)^2}{h_n(\mbba_{j-1}^\top \al_0)^2}\right)\mbba_{j-1}+o_p(1),
\end{align*}
and similarly we have
\begin{equation*}
\sqrt{T_n}(\bes^{\mathrm{cont}}-\beta_0)=\left(\int \frac{\mbbb^{\otimes2}(x)}{\mbba(x)^\top\al_0}\pi_0(dx)\right)^{-1}\frac{1}{\sqrt{T_n}}\sumj \frac{\D_j \Xc-h_nb_{j-1}}{\mbba_{j-1}^\top\al_0}\mbbb_{j-1}+o_p(1).
\end{equation*}
Theorem \ref{osan} now follows from applying Slutsky's lemma and Lemma \ref{ana} to these two equations.

\medskip

Before we proceed to the proof of Theorem \ref{Consistency} and Theorem \ref{Ae}, we comment on a virtual upper bound of $k_n$ and $N_{T_n}$.
By the Lindeberg-Feller theorem we have
\begin{equation*}
\frac{N_{T_n}-\lam T_n}{\sqrt{\lam T_n}} =\sumj\frac{\D_j N - \lam h_n}{\sqrt{\lam T_n}} \cil N(0,1),
\end{equation*}
so that for any positive nondecreasing sequence $(l_n)$ satisfying $\frac{l_n-\lam T_n}{\sqrt{\lam T_n}}\to\infty$, 
we have
\begin{equation}
P\left(N_{T_n}\geq l_n\right)=P\left(\frac{N_{T_n}-\lam T_n}{\sqrt{\lam T_n}}\geq\frac{l_n-\lam T_n}{\sqrt{\lam T_n}}\right)\to0;
\nn
\end{equation}
in particular, this implies that the probability of the event $\{N_{T_n}\geq (\lam+1)T_n\}$ is asymptotically negligible. Thus, we hereafter set $k_n\leq (\lam+1)T_n-1=O(T_n)$, and replace the event $\{N_{T_n}\geq k_n+1\}$ by $\{ k_n+1\leq N_{T_n}\leq (\lam+1)T_n\}$ without any mention.

\subsection{Proof of Theorem \ref{Consistency}}
Since it is easy to deduce that for a fixed $M>0$,
\begin{align*}
&P\left(\left\{|\sqrt{n}(\tilde{\al}_n^{k_n}-\al_0)|>M\right\}\cap\left\{\mathrm{JB}_n^{k_n}\leq\chi^2_q(2)\right\}\right)\\
&\leq P\left(\left\{|\sqrt{n}(\tilde{\al}_n^{k_n}-\al_0)|>M\right\}\cap\left\{1\leq N_{T_n}\leq k_n\right\}\right) + P\left(\left\{k_n+1\leq N_{T_n}\leq (\lam+1)T_n\right\}\cap\left\{\mathrm{JB}_n^{k_n}\leq\chi^2_q(2)\right\}\right)+o(1),
\end{align*}
the desired result follows if we show that for any $\ep>0$, there exist positive constants $M$ and $N\in\mbbn$ satisfying
\begin{align}
&\sup_{n\ge N}P\left(\left\{|\sqrt{n}(\tilde{\al}_n^{k_n}-\al_0)|>M\right\}\cap\left\{1\leq N_{T_n}\leq k_n\right\}\right)<\ep, \label{yu:C1}\\
&\sup_{n\ge N}P\left(\left\{k_n+1\leq N_{T_n}\leq (\lam+1)T_n\right\}\cap\left\{\mathrm{JB}_n^{k_n}\leq\chi^2_q(2)\right\}\right)<\ep. \label{yu:C2}
\end{align}
From now on, we separately prove them with introducing some fundamental lemmas.

\subsubsection{Proof of \eqref{yu:C1}}
We will write $\{\tau_i\}_{i\in\mbbn}$ for jump times of $N$, and $B_n$ for the event that the Poisson process $N$ does not have more than one jumps over all $[t_{j-1},t_j)$, $j=1,\dots,n$:
\begin{align}
B_{n} &:= \left\{^\exists i\in\mbbn, ^\exists j\in\{1,\dots,n\} \ \ \mbox{s.t.} \ \ \tau_i, \tau_{i+1}\in[t_{j-1},t_j)\right\}^c.
\nn
\end{align}

\begin{Lem} \label{ojump}
$P(B_n) = 1 - O(nh_n^2)$.
\end{Lem}

\begin{proof}
For each $i\ge 2$, the conditional distribution of $(\tau_1/T_n,\dots,\tau_i/T_n)$ given the event $\{N_{T_n}=i\}$ equals that of the order statistics $U_{(1)}\le\dots\le U_{(i)}$ of $k$ i.i.d. $(0,1)$-uniformly distributed random variables \cite[Proposition 3.4]{Sat99}. Moreover, each spacing $U_{(i+1)}-U_{(i)}$ admits the density $s\mapsto i(1-s)^{i-1}$, $0<s<1$, e.g. \cite{Pyk65}. Then,
\begin{align}
P(B_n^c) &= \sum_{i=2}^\infty P(N_{T_n}=i) P( B_n^c \, |\, N_{T_n}=i ) \nn\\
&\le \sum_{i=2}^\infty P(N_{T_n}=i) P\big( {}^\exists j\in\{2,\dots, i\} \ \ \mbox{s.t.} \ \ \tau_{i}-\tau_{i-1}<h_n \, \big| \, N_{T_n}=i \big) \nn\\
&\le \sum_{i=2}^\infty P(N_{T_n}=i) \times (i-1) \int_0^{1/n} i(1-s)^{i-1}ds \nn\\
&\lesssim \sum_{i=2}^\infty e^{-\lam T_n} \frac{(\lam T_n)^{i}}{(i-2)!}\times \frac1n \lesssim \frac{T_n^2}{n} =nh_n^2.
\nonumber
\end{align}
\end{proof}

Let
\begin{align}
C_{k,n} &:=\left\{ {}^\exists i\in\mbbn,\, {}^\exists j\in\{1,\dots,n\} \ \ \mbox{s.t.} \ \ \tau_i \in [t_{j-1},t_j) \ \ \mbox{and} \ \  j\notin\hat{\mcj}_n^{k} \right\}^c,
\nonumber
\end{align}
denote the event where all jumps up to time $T_n$ are correctly removed.
The next lemma shows the asymptotic negligibility of the failure-to-detection rate on the event $\left\{1\leq N_{T_n}\leq k_n \right\}\cap B_n$.

\begin{Lem} \label{cpick}
$P(C_{k_n,n}^{\,c}\cap \left\{1\leq N_{T_n}\leq k_n \right\}\cap B_n) \to 0$.
\label{proof_lem2}
\end{Lem}

\begin{proof}
Hereafter we use the following notations:
\begin{align*}
\mcd_n &= \{j\le n:\, ^\exists i, \text{ s.t. } \tau_i \in[t_{j-1},t_j)\}, \nn\\
\mcc_n &= \{1,\dots,n\}\setminus \mcd_n.
\end{align*}
Write
\begin{equation}
\eta_j=\frac{\D_j w}{\sqrt{h_n}}, \qquad j\le n.
\nonumber
\end{equation}
Recalling that the set $\hat{\mcj}_n^{k_{n}}$ of removed indices is constructed through picking up the first $k_n$-largest increments in magnitude, we have
\begin{align}
& P\Big( C_{k_n,n}^{\,c}\cap \left\{1\leq N_{T_n}\leq k_n \right\}\cap B_n \Big)
\nn\\
&\leq P(\{^\exists j'\in\mcd_n, j''\in\mcc_n \ \ \mbox{s.t.} \ \ |\D_{j'} X|<|\D_{j''} X|\}\cap\{1\leq N_{T_n}\leq k_n\}\cap B_n)
\nn\\
&\leq P\bigg(\bigg\{ {}^\exists j'\in\mcd_n, j''\in\mcc_n \ \ \mbox{s.t.} \ \ \inf_x|c(x)|\min_{1\leq j\leq N_{T_n}}|\xi_j|
\nn\\
&{}\qquad<\bigg|\int_{t_{j'-1}}^{t_{j'}}b_sds+\int_{t_{j'-1}}^{t_{j'}} a_sdw_s\bigg|+\bigg|\int_{t_{{j''}-1}}^{t_{j''}}b_sds+\int_{t_{{j''}-1}}^{t_{j''}} a_sdw_s\bigg|\bigg\}
\cap\{1\leq N_{T_n}\leq k_n\}\cap B_n\Bigg)
\nn\\
&\leq P\bigg( \bigg\{\inf_x|c(x)|\min_{1\leq j\leq N_{T_n}}|\xi_j| < 2\sqrt{h_n}\sup_x |a(x,\al_0)| \max_{1\leq j\leq n}|\eta_j|
\nn\\
&{}\qquad +2\max_{1\leq j\leq n}\bigg(\bigg|\intj b_sds\bigg|+\bigg|\intj(a_s-a_{j-1})dw_s\bigg|\bigg)\bigg\} \cap \{1\leq N_{T_n}\leq k_n\}\cap B_n\Bigg)
\nn\\
&\leq  P\Bigg(\left\{ \inf_x|c(x)|^2\min_{1\leq j\leq N_{T_n}}|\xi_j|^2<r_{1,n}+r_{2,n}\right\} \cap\{1\leq N_{T_n}\leq k_n\}\cap B_n\Bigg),
\label{A3-1}
\end{align}
where
\begin{align*}
&r_{1,n}:=8h_n\sup_x a^2(x,\al_0) \max_{1\leq j\leq n}|\eta_j|^2, \\
&r_{2,n}:=8\sumj\left\{\left(\intj b_sds\right)^2+\left(\intj (a_s-a_{j-1})dw_s\right)^2\right\}.
\end{align*}
From extreme value theory (cf. \cite[Table 3.4.4]{EmbKluMik97}), we have
\begin{equation}
\nn
\max_{1\leq j\leq n } |\eta_i|^2-\left(\log n- \frac{1}{2}\log\log n-\log\Gam\left(\frac{1}{2}\right)\right)=O_p(1).
\end{equation}
This together with Assumption \ref{Ascoef} and \eqref{hm:add-4} leads to
\begin{equation*}
r_{1,n}=O_p(h_n\log n) =O_p(nh_n^2).
\end{equation*}
Jensen's and Burkholder's inequalities together with \cite[Lemma 4.5]{Mas13-1} gives $E[r_{2,n}]\lesssim nh_n^2$,
so that
\begin{equation*}
r_{2,n}=O_p(nh_n^2).
\end{equation*}
Hence, for any $\ep\in(0,1)$, we can pick sufficiently large $N$ and $K$ such that for all $n\geq N$,
\begin{equation*}
P\left( r_{1,n} + r_{2,n} > K nh_n^2\right) <\ep.
\end{equation*}
Building on these estimates, $E[N_{T_n}]=\lam T_n$, and the independence between $N$ and $(\xi_j)$, we see that the upper bound in \eqref{A3-1} can be further bounded by
\begin{align}
& P\Bigg(\left\{\min_{1\leq j\leq N_{T_n}}|\xi_j|^2<\frac{K}{\inf_x|c(x)|^2}nh_n^2\right\}\cap\{1\leq N_{T_n}\leq k_n\}\cap B_n\Bigg) + \ep \nn\\
&\le \sum_{i=1}^{k_{n}}P\Bigg(\left\{\min_{1\leq j\leq i }|\xi_j|^2<\frac{K}{\inf_x|c(x)|^2}nh_n^2\right\}\cap\{N_{T_n}=i\}\Bigg) + \ep
\nn\\
&\le \sum_{i=1}^{k_{n}} i P\left(|\xi_1|^2<\frac{K}{\inf_x|c(x)|^2}nh_n^2\right) P\left( N_{T_n}=i\right) + \ep
\nn\\
&\lesssim T_n P\left(|\xi_1|<\frac{\sqrt{K}}{\inf_x|c(x)|}\sqrt{n}h_n\right) + \ep.
\label{A3-2}
\end{align}
Since the choice of $\ep$ is arbitrary, \eqref{hm:conseq-3} implies the desired result.
\end{proof}

\medskip


Let us introduce the event
\begin{equation}
G_{k_{n},n} := \left\{1\leq N_{T_n}\leq k_n\right\}\cap B_n\cap C_{k_n,n}.
\nonumber
\end{equation}
Thanks to Lemmas \ref{ojump} and \ref{cpick}, it follows that
\begin{align*}
&P\left(\left\{|\sqrt{n}(\tilde{\al}_n^{k_n}-\al_0)|>M\right\}\cap\left\{1\leq N_{T_n}\leq k_n\right\}\right)\\
&\leq P\left(\left\{|\sqrt{n}(\tilde{\al}_n^{k_n}-\al_0)|>M\right\}\cap 
G_{k_{n},n}
\right)+P(C_{k_n,n}^{\,c}\cap \left\{1\leq N_{T_n}\leq k_n \right\}\cap B_n)+P(B_n^c)\\
&=P\left(\left\{|\sqrt{n}(\tilde{\al}_n^{k_n}-\al_0)|>M\right\}\cap 
G_{k_{n},n}
\right)+o(1).
\end{align*}
Hence, in order to prove \eqref{yu:C1}, it suffices to show that for any $\ep>0$ there correspond sufficiently large $M>0$ and $N\in\mbbn$ for which
\begin{equation}\label{clse}
\sup_{n\ge N}P\left(\left\{|\sqrt{n}(\tilde{\al}_n^{k_n}-\al_0)|>M\right\}\cap 
G_{k_{n},n}
\right)<\ep.
\end{equation}
Since $\D_j X=\D_j \Xc$ for each $j\notin \hat{\mcj}_n^{k_n}$ on $G_{k_{n},n}$, we have
\begin{equation}
|\tilde{\al}_n^{k_n}-\al_0|\mbbi_{G_{k_{n},n}} \leq(|\kappa_{1,n}|+|\kappa_{2,n}|+|\kappa_{3,n}|)\mbbi_{G_{k_{n},n}},
\label{thm.4.7-p1}
\end{equation}
where
\begin{align*}
&\kappa_{1,n}:=\frac{1}{h_n}\left\{\left(\sum_{j\notin\hat{\mcj}^{k_n}_n}\mbba_{j-1}\mbba_{j-1}^\top\right)^{-1}-\left(\sumj\mbba_{j-1}\mbba_{j-1}^\top\right)^{-1}\right\}\sum_{j\notin\hat{\mcj}_n^{k_n}}\mbba_{j-1} (\D_j \Xc)^2,\\
&\kappa_{2,n}:=-\frac{1}{h_n}\left(\sumj\mbba_{j-1}\mbba_{j-1}^\top\right)^{-1}\sum_{j\in\hat{\mcj}_n^{k_n}}\mbba_{j-1} (\D_j \Xc)^2,\\
&\kappa_{3,n}:=\frac{1}{h_n}\left(\sumj\mbba_{j-1}\mbba_{j-1}^\top\right)^{-1}\left\{\sumj\mbba_{j-1} (\D_j \Xc)^2-h_n\left(\sumj\mbba_{j-1}\mbba_{j-1}^\top\right)\al_0\right\}.
\end{align*}
Below we look at these three terms separately.

\medskip

1. Evaluation of $\kappa_{1,n}$: From the ergodic theorem, we have
\begin{equation*}
\frac{1}{n}\sumj\mbba_{j-1}\mbba_{j-1}^\top\cip \int \mbba(x)(\mbba(x))^\top\pi_0(dx)>0,
\end{equation*}
so that $(\frac{1}{n}\sumj\mbba_{j-1}\mbba_{j-1}^\top)^{-1}=O_p(1)$ as a random sequence of matrices.
Since $\mbba(x)$ is bounded, we can also obtain
\begin{align*}
&\frac{1}{n}\sum_{j\notin\hat{\mcj}^{k_n}_n} \mbba_{j-1}\mbba_{j-1}^\top =\frac{1}{n}\sumj\mbba_{j-1}\mbba_{j-1}^\top-\frac{1}{n}\sum_{j\in\hat{\mcj}^{k_n}_n} \mbba_{j-1}\mbba_{j-1}^\top=\frac{1}{n}\sumj\mbba_{j-1}\mbba_{j-1}^\top+O_p\left(\frac{k_n}{n}\right),\\
&\left|\frac{1}{T_n}\sum_{j\notin\hat{\mcj}_n^{k_n}}\mbba_{j-1} (\D_j \Xc)^2\right|\lesssim \frac{1}{T_n}\sumj(\D_j \Xc)^2=O_p(1),
\end{align*}
from \eqref{ece}.
Hence it follows that
\begin{align}
&|\sqrt{n}\kappa_{1,n}|\mbbi_{G_{k_{n},n}} \nn\\
&\lesssim 
\left| \left( \frac{1}{n}\sumj\mbba_{j-1}\mbba_{j-1}^\top\right)^{-1} \right| \cdot 
\left|\sqrt{n}\left\{\left(\frac{1}{n}\sumj\mbba_{j-1}\mbba_{j-1}^\top\right)\left(\frac{1}{n}\sum_{j\notin\hat{\mcj}^{k_n}_n}\mbba_{j-1}\mbba_{j-1}^\top\right)^{-1}-I_{p_\al}\right\}\right|
\nn\\
&{}\qquad\times
\left(\frac{1}{T_n}\sumj (\D_j \Xc)^2\right) \nn\\
&=O_p\left(\frac{k_n}{\sqrt{n}}\right)=o_p(1).
\label{thm.4.7-p2}
\end{align}
\medskip

2. Evaluation of $\kappa_{2,n}$: Recall that $\eta_j:=\frac{\D_j w}{\sqrt{h_n}}$.
Under Assumption \ref{Ascoef}, we can derive from the estimates of $r_{1,n}$ and $r_{2,n}$ in the proof of Lemma \ref{cpick} that,
on $G_{k_{n},n}$,
\begin{align}
&\left|\frac{1}{T_n}\sum_{j\in\hat{\mcj}_n^{k_n}}\mbba_{j-1} (\D_j \Xc)^2\right|
\nonumber\\
&\lesssim \frac{1}{T_n} \left\{ \sumj \left[\left(\intj (a_s-a_{j-1})dw_s\right)^2+\left(\intj b_sds\right)^2\right]+k_nh_n\max_{1\leq j\leq n}|\eta_j|^2\right\}
\nn\\
& =O_p\left(\frac{1}{T_n}(nh_n^2 \vee k_n h_n \log n)\right)=O_p\left( h_n \vee \frac{k_n \log n}{n}\right).
\label{smoj}
\end{align}
Thus we get 
\begin{align}
|\sqrt{n}\kappa_{2,n}|\mbbi_{G_{k_{n},n}}
=O_p\left( \sqrt{n h_n^2} \vee \frac{k_n \log n}{\sqrt{n}}\right) = o_p(1).
\label{thm.4.7-p3}
\end{align}

\medskip
3. Evaluation of $\kappa_{3,n}$: 
From \eqref{sce}, \eqref{fce}, \eqref{ece}, and the martingale central limit theorem (see the proof of Lemma \ref{ana}), it follows that
\begin{align}
\sqrt{n}\kappa_{3,n} &= 
\left(\frac{1}{n}\sumj\mbba_{j-1}\mbba_{j-1}^\top\right)^{-1}
\frac{1}{\sqrt{n}}\sumj\mbba_{j-1} \left\{\left(\frac{\D_j \Xc}{\sqrt{h_n}}\right)^2 - \mbba_{j-1}^\top\al_0\right\} = O_p(1).
\label{thm.4.7-p4}
\end{align}

Substituting \eqref{thm.4.7-p2}, \eqref{thm.4.7-p3} and \eqref{thm.4.7-p4} into \eqref{thm.4.7-p1} now yields that
\begin{equation}
\left| \sqrt{n}\left(\tilde{\al}_n^{k_n}-\al_0\right)\right|\mbbi_{G_{k_{n},n}}=O_p(1),
\nonumber
\end{equation}
followed by \eqref{clse}.

\subsubsection{Proof of \eqref{yu:C2}}
Let
\begin{equation*}
D_{k_n,n}:=
\left\{\mcc_n\cap\hat{\mcj}_n^{k_n}=\emptyset\right\}.
\end{equation*}
Recall that the probability of the event $\{N_{T_n}\geq (\lam+1)T_n\}$ is asymptotically negligible and that without loss of generality, we can assume $k_n\leq (\lam+1)T_n-1$.
Then we get the following lemma.
\begin{Lem}
\label{hm:lem-1}
\begin{equation*}
P\left( D_{k_n,n}^c \cap\{k_n+1\leq N_{T_n}\leq (\lam+1)T_n\}\cap B_n\right)\to0.
\end{equation*}
\end{Lem}
\begin{proof}
The lemma can be shown in a quite similar way to Lemma \ref{proof_lem2}.
Letting $\ep$, $N$, and $K$ be the same as in the proof of Lemma \ref{proof_lem2}, as in \eqref{A3-1} and \eqref{A3-2}, we have for any $n\geq N$,
\begin{align*}
&P\left( D_{k_n,n}^c \cap\{k_n+1\leq N_{T_n}\leq (\lam+1)T_n\}\cap B_n\right)\\
&\le P\left(\left\{ {}^\exists j'\in\mcd_n,\, j''\in\mcc_n \ \ \mbox{s.t.} \ \ |\D_{j'} X|<|\D_{j''} X| \right\}\cap \{k_n+1\leq N_{T_n}\leq (\lam+1)T_n\}\cap B_n\right)
\nn\\
&\leq P\left(\left\{\min_{1\leq j\leq N_{T_n}}|\xi_j|^2<\frac{K}{\inf_x|c(x)|^2}nh_n^2\right\}\cap\{k_n+1\leq N_{T_n}\leq (\lam+1)T_n\}\cap B_n\right)\\
&\lesssim T_n P\left(|\xi_1|<\frac{\sqrt{K}}{\inf_x|c(x)|}\sqrt{n}h_n\right)+\ep = o(1) +\ep.
\end{align*}
Hence the proof is complete.
\end{proof}

Let
\begin{equation}
H_{k_{n},n}:=\left\{k_n+1\leq N_{T_n}\leq (\lam+1)T_n\right\}\cap B_n\cap D_{k_n,n}.
\nonumber
\end{equation}
Thanks to Lemma \ref{hm:lem-1}, \eqref{yu:C2} is led by
\begin{equation}
P\left( \left\{ \mathrm{JB}_n^{k_n}\leq\chi^2_q(2)\right\}\cap H_{k_n,n}\right) = o(1).
\label{YU:prob2}
\end{equation}

In view of the definition \eqref{yu:msnr}, \eqref{YU:prob2} follows on showing that for any $M>0$,
\begin{equation}
\label{YU:prob2+1}
P\left(\left\{\frac{1}{\sqrt{n}}\sum_{j\notin\hat{\mcj}_{n}^{k_n}}\left(\left(\hat{N}_j^k\right)^4-3\right)<M\right\}\cap H_{k_n,n}\right)=o(1);
\end{equation}
recall the notation $\hat{N}_j^k =(\hat{S}_n^k)^{-1/2}(\ep_j(\aes^k)-\bar{\hat{\ep}}_n^k)$;
for this purpose, we need to clarify asymptotic behaviors of $\bar{\hat{\ep}}_{n}^{k_n} \mbbi_{H_{k_{n},n}}$ and $\hat{S}_n^{k_n}\mbbi_{H_{k_{n},n}}$.

\medskip

Define the diverging real sequence $a_n$ by \eqref{hm:an_def}:
\begin{equation}
a_n=T_n^\eta \uparrow\infty
\nonumber
\end{equation}
with $\eta>0$ being small enough to ensure \eqref{hm:conseq-1} to \eqref{hm:conseq-2}.
Making $\eta>0$ smaller if necessary so that $\eta <1/4$, we may and do further suppose that
\begin{equation}
a_n\sqrt{h_n} \vee \frac{a_n^4 h_n}{nh_n^2} \to 0.
\nonumber
\end{equation}
First we will prove
\begin{equation}
\bar{\hat{\ep}}_{n}^{k_n} \mbbi_{H_{k_{n},n}} =O_p\left(a_n \sqrt{h_n} \right) = o_p(1).
\label{hm:4.8.p++1}
\end{equation}
Decompose $\bar{\hat{\ep}}_{n}^{k_n}$ as
\begin{equation*}
\bar{\hat{\ep}}_{n}^{k_n}=
\frac{1}{n-k_n}\sumj \ep_j(\aes^{k_n}) - \frac{1}{n-k_n}\sum_{j\in\hat{\mcj}_n^{k_n}}\ep_j(\aes^{k_n}).
\end{equation*}
Below we will look at the terms in the right-hand side separately.

Observe that
\begin{align}
\frac{1}{\sqrt{h_n}} \left|\frac{1}{n-k_n}\sumj \ep_j(\aes^{k_n})\right|
&\lesssim
\left|\frac{1}{n} \sumj\frac{1}{a_{j-1} (\aes^{k_n})} \frac{1}{h_n}\intj (b_s+c_s\lam E[\xi_1])ds \right| \nn\\
&{}\qquad 
+\left| \sumj \frac{1}{T_n} \frac{1}{a_{j-1} (\aes^{k_n})} \left(\intj a_sdw_s \right)\right| \nn\\
&{}\qquad
+\left| \sumj \frac{1}{T_n} \frac{1}{a_{j-1} (\aes^{k_n})} \left(\intj c_{s-}d\tilde{J}_s \right) \right|,
\label{hm:jbh-1}
\end{align}
where $\tilde{J}_t:=J_t-\lam E[\xi_1]t$.
Obviously the first term in \eqref{hm:jbh-1} is $O_p(1)$, and we are going to show that the remaining two terms are $o_p(1)$.
To achieve this, under the present assumptions it suffices to prove the following claim:
let $\pi(x,\al)$ be a bounded real-valued function on $\mbbr\times\Theta_\al$ such that $|\pi(x,\al)-\pi(x,\al')| \lesssim |\al-\al'|$ for each $x\in\mbbr$ and $\al,\al'\in\bar{\Theta}_\al$,
and consider the random $\mcc^1(\bar{\Theta}_\al)$-functions
\begin{align*}
&F_{1,n}(\al):= \sumj \frac{1}{T_n} \pi_{j-1}(\al) \intj a_{s}dw_s,\\
&F_{2,n}(\al):= \sumj \frac{1}{T_n} \pi_{j-1}(\al) \intj c_{s-}d\tilde{J}_s.
\nonumber
\end{align*}
Then we claim
\begin{equation}
\sup_{\al\in\bar{\Theta}_\al}|F_{1,n}(\al)| = o_p(1), \quad \sup_{\al\in\bar{\Theta}_\al}|F_{2,n}(\al)|=o_p(1)
\label{hm:4.8.p++2.5}
\end{equation}
To show \eqref{hm:4.8.p++2.5}, we note that $F_{1,n}(\al)\cip 0$ and $F_{2,n}(\al)\cip0$ for each $\al$, which follows on applying \cite[Lemma 9]{GenJac93}.
Concerning $F_{1,n}$, by making use of Burkholder's inequality and Jensen's inequality, it is easy to deduce that for each integer $q>(p_\al\vee 2)$, $E\left[|F_{1,n}(\al)|^q\right]\lesssim 1$ and $E\left[ |F_{1,n}(\al)-F_{1,n}(\al')|^q\right]\lesssim |\al-\al'|^q$.
As for $F_{2,n}$, proceeding as in \cite[Eq.(4.14)]{Mas13-1} we can verify that for each integer $q>(p_\al \vee 2)$:
letting
\begin{align*}
\chi_j(t):=
\begin{cases}
1 & t\in(t_{j-1},t_j],\\
0 & \text{otherwise},
\end{cases}
\end{align*}
we have
\begin{align*}
E\left[|F_{2,n}(\al)|^q\right]&\lesssim T_n^{-q} E\left\{\left|
\int_0^{T_n} \left(\sumj \chi_j(s) \pi_{j-1}(\al)c_{s-} \right) d\tilde{J}_s \right|^q \right\} \nn\\
&\lesssim T_n^{-q/2} \frac{1}{T_n}\int_0^{T_n} \sumj \chi_j(s) E\left[ |c_{s-}|^q \right] ds \nn\\
&\lesssim1,
\end{align*}
and
\begin{align}
E\left[ |F_{2,n}(\al)-F_{2,n}(\al')|^q\right]
&\lesssim T_n^{-q} E\left\{\left|
\int_0^{T_n} \left(\sumj \chi_j(s) (\pi_{j-1}(\al)-\pi_{j-1}(\al'))c_{s-} \right) d\tilde{J}_s \right|^q \right\} \nn\\
&\lesssim T_n^{-q/2} \frac{1}{T_n}\int_0^{T_n} \sumj \chi_j(s) E\left\{ |\pi_{j-1}(\al)-\pi_{j-1}(\al')|^q |c_{s-}|^q \right\} ds \nn\\
&\lesssim T_n^{-q/2} |\al-\al'|^q \lesssim |\al-\al'|^q.
\nonumber
\end{align}
Hence the Kolmogorov criterion (cf. \cite[Theorem 1.4.7]{Kun97}) concludes the tightness of $\{F_{1,n}(\cdot)\}_n$ and $\{F_{2,n}(\cdot)\}_n$ in the space $\mcc(\bar{\Theta}_\al)$ (equipped with the uniform metric), from which \eqref{hm:4.8.p++2.5} follows.
We thus conclude
\begin{equation}
\frac{1}{n-k_n}\sumj \ep_j(\aes^{k_n}) = O_p\left(\sqrt{h_n}\right).
\label{hm:4.8.p++2}
\end{equation}


Next, it follows from Assumptions \ref{Ascoef} and \ref{Sampling} that
\begin{align}\label{esresi}
\max_{1\leq j\leq n}\ep_j^2(\aes^{k_n})\mbbi_{H_{k_{n},n}}
&\lesssim \frac{1}{h_n}\max_{1\leq j\leq n}(\D_j X)^2\mbbi_{H_{k_{n},n}}\nn\\
&\lesssim \frac{1}{h_n}\sumj \left\{\left(\intj b_sds\right)^2+\left(\intj (a_s-a_{j-1})dw_s\right)^2\right\} \nn\\
&{}\qquad+\max_{1\leq j\leq n} \eta_j^2+\frac{1}{h_n}\max_{1\leq j\leq \lf(\lam+1)T_n\rf} \xi_j^2\nn\\
& \lesssim O_p(T_n) + O_p(\log n) + O_p\left(\frac{a_n^2}{h_n}\right) \nn\\
& =O_p\left(\frac{a_n^2}{h_n}\left(\frac{T_n h_n}{a_n^2} +1\right)\right)=O_p\left(\frac{a_n^2}{h_n}\right).
\end{align}
This gives
\begin{align}
\left|\frac{1}{n-k_n}\sum_{j\in\hat{\mcj}_n^{k_n}} \ep_j(\aes^{k_n})\right|\mbbi_{H_{k_{n},n}}
&\lesssim \frac{k_n}{n}\sqrt{\max_{1\leq j\leq n}\ep_j^2(\aes^{k_n})}\,\mbbi_{H_{k_{n},n}} \nn\\
&= O_p\left(\frac{k_n a_n}{n\sqrt{h_n}}\right) = O_p\left(a_n\sqrt{h_n} \right),
\label{hm:4.8.p++3}
\end{align}
and \eqref{hm:4.8.p++1} follows from \eqref{hm:4.8.p++2} and \eqref{hm:4.8.p++3}.

\medskip

Next we look at $\hat{S}_n^{k_n}\mbbi_{H_{k_{n},n}}$. Note that \eqref{hm:4.8.p++1} under Assumption \ref{Asjsize} entails
\begin{align}
\hat{S}_n^{k_n}\mbbi_{H_{k_{n},n}}
&=\frac{1}{n-k_n}\sum_{j\notin\hat{\mcj}_n^{k_n}} \ep^2_j(\aes^{k_n})\mbbi_{H_{k_{n},n}} + o_p(1).
\label{hm:4.8.p++4}
\end{align}
From Assumption \ref{Ascoef}, the following relation holds:
\begin{equation}\label{rela}
\frac{1}{T_n}\sum_{j\notin\hat{\mcj}_n^{k_n}}(\D_jX)^2\lesssim\frac{1}{n-k_n}\sum_{j\notin\hat{\mcj}_n^{k_n}} \ep^2_j(\aes^{k_n})\lesssim \frac{1}{T_n}\sumj (\D_jX)^2. 
\end{equation}
From Cauchy-Schwarz inequality, Burkholder's inequality and \cite[Lemma 4.5]{Mas13-1}, we derive
\begin{align*}
E\left[\left(\D_j X\right)^2\right]
&=E\Bigg[\bigg(\intj(a_s-a_{j-1})dw_s+\intj (b_s+\lam E(\xi_1) c_s)ds \nn\\
&{}\qquad +\intj(c_{s-}-c_{j-1})d\tilde{J}_s+a_{j-1}\D_jw+c_{j-1}\D_j\tilde{J}\bigg)^2\Bigg]\\
&=E\left[\left(a_{j-1}\D_j w+c_{j-1}\D_j \tilde{J}\right)^2\right]+O\left(h_n^{\frac{3}{2}}\right)\\
&\lesssim h_n.
\end{align*}
Hence the rightmost side in \eqref{rela} is $O_p(1)$.
In a similar manner through Cauchy-Schwarz and Burkholder's inequalities, we have
\begin{align*}
\frac{1}{T_n}\sum_{j\notin\hat{\mcj}_n^{k_n}}(\D_jX)^2
&=\frac{1}{T_n}\sum_{j\notin\hat{\mcj}_n^{k_n}}\left(a_{j-1}\D_j w+c_{j-1}\D_j \tilde{J}\right)^2+O_p\left(\sqrt{h_n}\right)\\
&=\frac{1}{T_n}\sum_{j\notin\hat{\mcj}_n^{k_n}} \left(c_{j-1}\D_j \tilde{J}\right)^2
+\frac{1}{T_n}\sumj\left\{\left(a_{j-1}\D_j w\right)^2+2a_{j-1}c_{j-1}\D_j w\D_j \tilde{J}\right\}\\
&{}\qquad-\frac{1}{T_n}\sum_{j\in\hat{\mcj}_n^{k_n}} \left\{\left(a_{j-1}\D_j w\right)^2+2a_{j-1}c_{j-1}\D_j w\D_j \tilde{J}\right\} + o_p(1) \nn\\
&\ge \frac{1}{T_n}\sumj\left\{\left(a_{j-1}\D_j w\right)^2+2a_{j-1}c_{j-1}\D_j w\D_j \tilde{J}\right\}\\
&{}\qquad-\frac{1}{T_n}\sum_{j\in\hat{\mcj}_n^{k_n}} \left\{\left(a_{j-1}\D_j w\right)^2+2a_{j-1}c_{j-1}\D_j w\D_j \tilde{J}\right\} + o_p(1) \nn\\
&=: L_{n} - \hat{L}^{k_n}_n + o_p(1).
\end{align*}
The independence between $w$ and $J$, \cite[Lemma 9]{GenJac93}, and the ergodic theorem yield that
\begin{equation*}
L_n \cip \int a^2(x,\al_0)\pi_0(dx)>0.
\end{equation*}
In a similar manner to \eqref{esresi}, Assumption \ref{Asjsize} implies that
\begin{align*}
| \hat{L}^{k_n}_n | \mbbi_{H_{k_{n},n}} \le  O_p\left(\frac{k_n \log n}{n}\right) + O_p\left(a_n\sqrt{h_n\log n}\right) = o_p(1).
\end{align*}
Summarizing the last three displays leads to
\begin{equation}
\frac{1}{T_n}\sum_{j\notin\hat{\mcj}_n^{k_n}}(\D_jX)^2 \ge \int a^2(x,\al_0)\pi_0(dx) + o_p(1) - \mbbi_{H_{k_{n},n}^c}| \hat{L}^{k_n}_n |
\nonumber
\end{equation}
Fix an arbitrary $\ep>0$.
By the last display combined with \eqref{hm:4.8.p++1}, \eqref{hm:4.8.p++4} and \eqref{rela},
we can pick a positive constant $K=K(\ep)>1$ and a positive integer $N=N(\ep)$ such that
\begin{equation*}
\sup_{n\ge N}P\left[\left(\left\{\hat{S}_n^{k_n}<\frac{1}{K}\right\}\cup\left\{\hat{S}_n^{k_n}>K\right\} \cup 
\left\{|\bar{\hat{\ep}}_n^{k_n}|>K a_n \sqrt{h_n} \right\} \right) \cap H_{k_{n},n} \right] < \ep.
\end{equation*}
Therefore, to conclude \eqref{YU:prob2+1} we may and do focus on the event
\begin{equation*}
F_{k_n,n,\ep} := \left\{ \frac{1}{K} \le \hat{S}_n^{k_n} \le K\right\} \cap\left\{|\bar{\hat{\ep}}_n^{k_n}|\leq K a_n \sqrt{h_n}  \right\} \cap H_{k_{n},n};
\end{equation*}
we note that on $F_{k_n,n,\ep}$ there remain jumps (not removed), its number of pieces being at least one jump.

From Assumption \ref{Ascoef},
\begin{align}
|h_n^{-1/2}\D_j X|^k &\lesssim |h_n^{-1/2}\D_j \Xc|^k + |h_n^{-1/2}\D_j J|^k, \nn\\
|h_n^{-1/2}\D_j X|^k &\gtrsim |h_n^{-1/2}\D_j J|^k - |h_n^{-1/2}\D_j \Xc|^k
\nonumber
\end{align}
for $k>0$ hold on $B_n$. 
With these together with \eqref{ece}, writing positive constants $C=C(a,c)$ possibly varying from line to line,
we have
\begin{align}
&\frac{1}{\sqrt{n}}\sum_{j\notin\hat{\mcj}_{n}^{k_n}}\left(\left(\hat{N}_j^k\right)^4-3\right)\mbbi_{F_{k_n,n,\ep}} \nn\\
&\gtrsim\sqrt{n}\left\{\frac{1}{n}\sum_{j\notin\hat{\mcj}_{n}^{k_n}}\left(\left|\frac{\D_j X}{\sqrt{h_n}}\right|^4
-C a_n\sqrt{h_n}  \left|\frac{\D_j X}{\sqrt{h_n}}\right|^3\right) + O_p(1)\right\} \mbbi_{F_{k_n,n,\ep}}\nn\\
&\gtrsim \sqrt{n}\left\{
\frac{1}{nh_n^2}\sum_{j\notin\hat{\mcj}_{n}^{k_n}} \left( \left(\D_j J\right)^4
-C|\D_j J|^3 a_n h_n  \right) +O_p(1) \right\} \mbbi_{F_{k_n,n,\ep}}.
\nn\\
&\gtrsim \sqrt{n}\left\{
\frac{1}{nh_n^2}\sum_{j\notin\hat{\mcj}_{n}^{k_n}} \left(
\min_{1\leq j\leq \lf(\lam+1)T_n\rf} |\xi_j|^4 - C h_na_n \max_{1\leq j\leq \lf(\lam+1)T_n\rf}|\xi_j|^3
\right) +O_p(1) \right\} \mbbi_{F_{k_n,n,\ep}}.
\label{hm:addadd-1}
\end{align}
Since $a_n^4 h_n /(nh_n^2) = n^{-\del''}\to 0$ for some $\del''>0$, we have for every $M,M'>0$
\begin{align}
& P\left(
\min_{1\leq j\leq \lf(\lam+1)T_n\rf} |\xi_j|^4 - M h_na_n \max_{1\leq j\leq \lf(\lam+1)T_n\rf}|\xi_j|^3 \ge M' nh_n^2
\right) \nn\\
&= P\left( \min_{1\leq j\leq \lf(\lam+1)T_n\rf} |\xi_j|^4 \gtrsim O_p(a_n^4 h_n) + nh_n^2 \right) \nn\\
&= P\left( \frac{\min_{1\leq j\leq \lf(\lam+1)T_n\rf} |\xi_j|}{(nh_n^2)^{1/4}} \gtrsim 1 \right) + o(1) \nn\\
&= \left\{ 1-P\left(|\xi_1| \lesssim (nh_n^2)^{1/4}\right)\right\}^{\lf(\lam+1)T_n\rf} + o(1).
\nonumber
\end{align}
The last probability tends to $1$ since
\begin{align}
T_n P\left(|\xi_1| \lesssim (nh_n^2)^{1/4}\right) &\lesssim n^{1-\kappa + (1-2\kappa)s/4} \to 0,
\nonumber
\end{align}
by \eqref{Asjsize}.
Recalling that on the event $F_{k_n,n,\ep}$ there is at least one jump over $[0,T_n]$ which is yet to be removed, we can continue \eqref{hm:addadd-1} as follows:
on an event whose probability gets arbitrarily close to $1$ as $n\to\infty$,
\begin{align}
&\gtrsim \sqrt{n}\bigg\{ \frac{1}{nh_n^2}\bigg(\min_{1\leq j\leq \lf(\lam+1)T_n\rf} |\xi_j|^4 
-C h_n a_n \max_{1\leq j\leq \lf(\lam+1)T_n\rf}|\xi_j|^3
\bigg)
+O_p(1)
\bigg\}\mbbi_{F_{k_n,n,\ep}}
\nn\\
&\gtrsim \sqrt{n}\left( M' + O_p(1) \right) \mbbi_{F_{k_n,n,\ep}} \nn\\
&\gtrsim \sqrt{n} \mbbi_{F_{k_n,n,\ep}}.
\nn
\end{align}
This entails \eqref{YU:prob2+1}, hence \eqref{YU:prob2} as well. 
\medskip

\subsection{Proof of Theorem \ref{Ae}}



We will complete the proof of Theorem \ref{Ae} by showing
\begin{align}
& P\left(\left\{\left|\sqrt{n}(\aes^{k_n}-\aes^{k_n, \mathrm{cont}})\right|\vee\left|\sqrt{T_n}(\bes^{k_n}-\bes^{k_n,\mathrm{cont}})\right|>\ep\right\}\cap G_{k_{n},n}
\right) = o(1), \label{YU:prob1}
\end{align}
Indeed, we can deduce from Lemmas \ref{ojump}, \ref{cpick}, and \ref{hm:lem-1} that for any $\ep>0$, the probabillity
\begin{align}
&P\left(\left\{\left|\sqrt{n}(\aes^{k_n}- \aes^{k_n, \mathrm{cont}} )\right|\vee\left|\sqrt{T_n}(\bes^{k_n}-\bes^{k_n,\mathrm{cont}})\right|>\ep\right\}\cap \left\{ \mathrm{JB}_n^{k_n}\leq\chi^2_q(2)\right\}\right)
\nn
\end{align}
can be bounded from above by the sum of the two probabilities given in \eqref{YU:prob1} and \eqref{YU:prob2}, plus an $o(1)$ term.
Recall that for any $j\notin \mcj_n^{k_n}$, $\D_j X=\D_j \Xc$ on $G_{k_{n},n}$.
Making use of Assumption \ref{Ascoef}, \eqref{smoj}, and a similar argument to the proof of Theorem \ref{Consistency}, we get
\begin{align*}
&\left|\sqrt{n}(\aes^{k_n}-\aes^{k_n, \mathrm{cont}})\mbbi_{G_{k_{n},n}}\right|\\
&\leq \left|\left(\frac{1}{n}\sum_{j\notin\hat{\mcj}_n^{k_n}}\frac{\mbba_{j-1}\mbba_{j-1}^\top}{\mbba_{j-1}^\top\tilde{\al}_n^{k_n}}\right)^{-1}\left\{\frac{1}{\sqrt{n}}\sum_{j\in\hat{\mcj}_n^{k_n}} \left(\frac{1}{\mbba_{j-1}^\top\tilde{\al}_n^{k_n}}-\frac{(\D_j \Xc)^2}{h_n(\mbba_{j-1}^\top \tilde{\al}_n^{k_n})^2}\right)\mbba_{j-1}\right\}\right|
\mbbi_{G_{k_{n},n}}\\
&{}\quad+\left|\left(\frac{1}{n}\sumj\frac{\mbba_{j-1}\mbba_{j-1}^\top}{\mbba_{j-1}^\top\tilde{\al}_n^{k_n}}\right)^{-1}\right|\left|\left(\frac{1}{n}\sumj\frac{\mbba_{j-1}\mbba_{j-1}^\top}{\mbba_{j-1}^\top\tilde{\al}_n^{k_n}}\right)\left(\frac{1}{n}\sum_{j\notin\hat{\mcj}_n^{k_n}}\frac{\mbba_{j-1}\mbba_{j-1}^\top}{\mbba_{j-1}^\top\tilde{\al}_n^{k_n}}\right)^{-1}-I_{p_\al}\right|\\
&{}\quad\times\left|\frac{1}{\sqrt{n}}\sumj\left(\frac{1}{\mbba_{j-1}^\top\tilde{\al}_n^{k_n}}-\frac{(\D_j \Xc)^2}{h_n(\mbba_{j-1}^\top \tilde{\al}_n^{k_n})^2}\right)\mbba_{j-1}\right|\mbbi_{G_{k_{n},n}}\\
&=
O_p(1) \cdot \left\{O_p\left(\frac{k_n}{\sqrt{n}}\right) + O_p\left( \sqrt{n h_n^2} \vee \frac{k_n \log n}{\sqrt{n}}\right)\right\}
+ O_p(1) \cdot o_p(1) \cdot O_p(1)
\nn\\
&=o_p(1).
\end{align*}
Next, to deduce $\sqrt{T_n}(\bes^{k_n}-\bes^{k_n,\mathrm{cont}}) \mbbi_{G_{k_{n},n}} = o_p(1)$, it suffices to prove
\begin{align}
& \frac{1}{n}\sumj \frac{\mbbb_{j-1}\mbbb_{j-1}^\top}{\mbba_{j-1}^\top\aes^{k_n}}\mbbi_{G_{k_{n},n}}
=\frac{1}{n}\sum_{j\notin\hat{\mcj}_n^{k_n}} \frac{\mbbb_{j-1}\mbbb_{j-1}^\top}{\mbba_{j-1}^\top\aes^{k_n}}
\mbbi_{G_{k_{n},n}}+o_p(1), \label{YU:exp1} \\
&\left|\frac{1}{\sqrt{T_n}}\sum_{j\in\hat{\mcj}_n^{k_n}} \frac{\D_j \Xc-h_nb_{j-1}}{\mbba_{j-1}^\top\aes^{k_n}}\mbbb_{j-1}\right|\nn \\
&=\left|\frac{1}{\sqrt{T_n}}\sum_{j\in\hat{\mcj}_n^{k_n}} \frac{\intj (b_s-b_{j-1})ds+\intj (a_s-a_{j-1})dw_s+a_{j-1}\D_j w}{\mbba_{j-1}^\top\aes^{k_n}}\mbbb_{j-1}\right|=o_p(1).
\label{YU:exp2} 
\end{align}
Then, as in the proof of Theorem \ref{osan} we see that
\begin{align*}
&\left|\sqrt{T_n}(\bes^{k_n}-\bes^{k_n,\mathrm{cont}})\mbbi_{G_{k_n,n}}\right|\\
&\le \left|\left(\frac{1}{n}\sum_{j\notin\hat{\mcj}_n^{k_n}} \frac{\mbbb_{j-1}\mbbb_{j-1}^\top}{\mbba_{j-1}^\top\aes^{k_n}}\right)^{-1} \frac{1}{\sqrt{T_n}}\sum_{j\in\hat{\mcj}_n^{k_n}} \frac{\D_j \Xc-h_nb_{j-1}}{\mbba_{j-1}^\top\aes^{k_n}}\mbbb_{j-1}\right|
\mbbi_{G_{k_n,n}}\\
&{}\qquad+\left|\left(\frac{1}{n}\sumj \frac{\mbbb_{j-1}\mbbb_{j-1}^\top}{\mbba_{j-1}^\top\aes^{k_n}}\right)^{-1}\right|\left|\left(\frac{1}{n}\sumj \frac{\mbbb_{j-1}\mbbb_{j-1}^\top}{\mbba_{j-1}^\top\aes^{k_n}}\right)\left(\frac{1}{n}\sum_{j\notin\hat{\mcj}_n^{k_n}} \frac{\mbbb_{j-1}\mbbb_{j-1}^\top}{\mbba_{j-1}^\top\aes^{k_n}}\right)^{-1}-I_{p_\beta}\right|\\
&{}\qquad\times\left|\frac{1}{\sqrt{T_n}}\sumj \frac{\D_j \Xc-h_nb_{j-1}}{\mbba_{j-1}^\top\aes^{k_n}}\mbbb_{j-1}\right|
\mbbi_{G_{k_n,n}}\\
&=o_p(1).
\end{align*}

By It\^{o}'s formula,
\begin{align}\label{YU:supesti}
&\left|\frac{1}{n}\sum_{j\in\hat{\mcj}_n^{k_n}} \frac{\mbbb_{j-1}\mbbb_{j-1}^\top}{\mbba_{j-1}^\top\aes^{k_n}}\right|
\mbbi_{G_{k_{n},n}}\nn\\
&\lesssim \frac{k_n}{n}\left(1+\sup_{0\leq t\leq T_n} X_t^2\right)\nn\\
&=\frac{k_n}{n}\sup_{0\leq t\leq T_n} \left\{1+X_0^2+2\int_0^t X_{s-}dX_s+\int_0^t a_s^2ds+\sum_{0<s\leq t} (\D_s X)^2\right\}\nn\\
&\lesssim\frac{k_n}{n}\Bigg\{1+X_0^2+\int_0^{T_n}\left(a_s^2+|X_sb_s|+c_s^2 +\left|X_{s}c_s\lam E[\xi_1]\right|\right)ds\nn\\
&+\sup_{0\leq t\leq T_n}\left|\int_0^tX_sa_sdw_s+\int_0^t\int_\mbbr \left(c_{s-}^2z^2+X_{s-}c_{s-}z\right)\tilde{N}(ds,dz)\right|\Bigg\},
\end{align}
where $\tilde{N}(\cdot,\cdot)$ denotes the compensated Poisson random measure associated with $J$.
Applying Assumption \ref{Moments} and Burkholder's inequality to the last term, we get 
\begin{equation*}
\frac{1}{n}\sum_{j\in\hat{\mcj}_n^{k_n}} \frac{\mbbb_{j-1}\mbbb_{j-1}^\top}{\mbba_{j-1}^\top\aes^{k_n}}\mbbi_{G_{k_{n},n}}
=O_p(h_n k_n)=o_p\left(\frac{\sqrt{nh_n^2}}{\log n} \right)=o_p(1),
\end{equation*}
hence \eqref{YU:exp1}.

Utilizing the Lipschitz continuity of $b$ and \cite[Lemma 4.5]{Mas13-1}, we have
\begin{align*}
E\left[\left|\frac{1}{\sqrt{T_n}}\sum_{j\in\hat{\mcj}_n^{k_n}} \frac{\intj (b_s-b_{j-1})ds}{\mbba_{j-1}^\top\aes^{k_n}}\mbbb_{j-1}\right|\right]&\lesssim \frac{1}{\sqrt{T_n}} \sumj\intj E\left[|(b_s-b_{j-1})\mbbb_{j-1}|\right]ds\\
&=O_p\left(\sqrt{nh_n^2}\right)=o_p(1).
\end{align*}
From the elementary inequality 
\begin{equation}\label{YU:ineq}
|x|\leq \frac12 \left(C+\frac{|x|^2}{C}\right),
\end{equation}
for any positive constant $C$ and real number $x$, we get
\begin{align*}
&\left|\frac{1}{\sqrt{T_n}}\sum_{j\in\hat{\mcj}_n^{k_n}}\frac{\intj (a_s-a_{j-1})dw_s}{\mbba_{j-1}^\top\aes^{k_n}}\mbbb_{j-1}\right|\\
&\lesssim 
\frac{1}{\sqrt{T_n}}\sum_{j\in\hat{\mcj}_n^{k_n}}
\left|\intj (a_s-a_{j-1})dw_s\mbbb_{j-1}\right|
\nn\\
&\lesssim\frac{1}{\sqrt{T_n}}\sum_{j\in\hat{\mcj}_n^{k_n}}\left\{\frac{\sqrt{T_n}}{k_n(\log n)^2}+\frac{k_n(\log n)^2}{\sqrt{T_n}}\left|\intj (a_s-a_{j-1})dw_s\mbbb_{j-1}\right|^2\right\}\\
&\lesssim
\frac{1}{(\log n)^2}
+\frac{1}{n}\sumj \left|\frac{1}{h_n}\intj (a_s-a_{j-1})dw_s\mbbb_{j-1}\right|^2 nh_n^2(\log n)^2 \\
&\lesssim \frac{1}{(\log n)^2}+O_p\left(nh_n^2(\log n)^2\right)=o_p(1).
\end{align*}
Here we used the condition $k_n\le (\lam+1)T_n-1$ and Burkholder's inequality.
By means of It\^{o}'s formula as in \eqref{YU:supesti} under Assumption \ref{Moments}, we obtain for any $q>2$
\begin{equation*}
E\left[\sup_{0\leq t\leq T_n} |X_t|^q\right]=O\left(T_n\right).
\end{equation*}
This combined with Jensen's inequality shows that
\begin{equation}\label{YU:supesti2}
E\left[\sup_{0\leq t\leq T_n} |X_t|^r \right]=O(T_n^\ep)
\end{equation}
for any $\ep>0$ and $r>0$.
Now, with the $\del\in(0,1)$ given in \eqref{hm:add-4} (a consequence of Assumption \ref{Sampling}), we set $\ep=\frac{\del}{3}$ and $\del'=\frac{4}{3}\del$.
Then, making use of \eqref{YU:supesti2} with an application of \eqref{YU:ineq}, we have
\begin{align*}
&\left|\frac{1}{\sqrt{T_n}}\sum_{j\in\hat{\mcj}_n^{k_n}} \frac{a_{j-1}\D_j w}{\mbba_{j-1}^\top\aes^{k_n}}\mbbb_{j-1}\right|\\
&\lesssim\frac{\max_{1\leq j \leq n}\left|\mbbb_{j-1}\right|}{\sqrt{T_n}}\sum_{j\in\hat{\mcj}_n^{k_n}}\left\{\frac{T_n^{\frac{1-\del'}{2}}}{k_n}+\frac{k_n}{T_n^{\frac{1-\del'}{2}}}(\D_j w)^2\right\}\\
&\lesssim  O_p\left(T_n^{-\frac{\del'}{2}+\ep} \vee T_n^{1+\ep+\frac{\del'}{2}}h_n\log n\right)\\
&=O_p\left(T_n^{-\frac{\del}{3}}\vee n^{1+\del}h_n^{2+\del}\log n\right)=o_p(1),
\end{align*}
thus concluding \eqref{YU:exp2}.

\medskip

\bigskip

\bibliographystyle{abbrv} 

\end{document}